\documentclass[12pt, onecolumn]{IEEEtran}
\usepackage{amsmath}
\usepackage{amssymb}
\usepackage{mathrsfs}
\usepackage{cite}
\usepackage{color}
\usepackage{epsfig}
\usepackage{epsf}
\usepackage{theorem}
\usepackage{graphics}
\usepackage[active]{srcltx}

\renewcommand{\baselinestretch}{1.5}

\newtheorem{thm}{Theorem}
\newtheorem{lem}{Lemma}

\newtheorem{defi}{Definition}

\begin{document}
\title{On the Delay-Throughput Tradeoff in Distributed Wireless Networks\\
\thanks{$^{*}$ This work is financially supported by Nortel Networks and the corresponding matching funds by the Natural Sciences
and Engineering Research Council of Canada (NSERC), and Ontario Centers of Excellence (OCE).}
\thanks{$^{*}$
The material in this paper was presented in part at the IEEE International Symposium on Information Theory (ISIT),
Nice, France, June 24-29, 2007 \cite{JamshidISIT2}.}}

\author{\small Jamshid Abouei, Alireza Bayesteh, and Amir K. Khandani \\
\small Coding and Signal Transmission Laboratory (www.cst.uwaterloo.ca)\\
Department of Electrical and Computer Engineering, University of Waterloo\\
Waterloo, Ontario, Canada, N2L 3G1 \\
Tel: 519-884-8552, Fax: 519-888-4338\\
Emails: \{jabouei, alireza, khandani\}@cst.uwaterloo.ca}

\maketitle

\markboth{\small Submitted to IEEE Transactions on Information
Theory, October 2009}{}

\begin{abstract}
This paper deals with the delay-throughput analysis of a single-hop
wireless network with $n$ transmitter/receiver pairs. All channels
are assumed to be block Rayleigh fading with shadowing, described by
parameters $(\alpha,\varpi)$, where $\alpha$ denotes the probability
of shadowing and $\varpi$ represents the average cross-link gains.
The analysis relies on the distributed \textit{on-off power
allocation strategy} (i.e., links with a direct channel gain above a
certain threshold transmit at full power and the rest remain silent)
for the deterministic and stochastic packet arrival processes. It is
also assumed that each transmitter has a buffer size of one packet
and dropping occurs once a packet arrives in the buffer while the
previous packet has not been served. In the first part of the paper,
we define a new notion of performance in the network, called
\textit{effective throughput}, which captures the effect of arrival
process in the network throughput, and maximize it for different
cases of packet arrival process. It is proved that the effective
throughput of the network asymptotically scales as $\frac{\log
n}{\hat{\alpha}}$, with $\hat{\alpha} \triangleq \alpha \varpi$,
regardless of the packet arrival process. In the second part of the
paper, we present the delay characteristics of the underlying
network in terms of the packet dropping probability. We derive the
sufficient conditions in the asymptotic case of $n \to \infty$ such
that the packet dropping probability tend to zero, while achieving
the maximum effective throughput of the network. Finally, we study
the trade-off between the effective throughput, delay, and packet
dropping probability of the network for different packet arrival
processes. In particular, we determine how much degradation will be
enforced in the throughput by introducing the aforementioned
constraints.
\end{abstract}

\begin{center}
\end{center}
\vskip 0.8cm

\section{Introduction}

As the demand for higher data rates increases, effective resource
allocation emerges as the primary issue in wireless networks in
order to satisfy Quality of Service (QoS) requirements. Central to
the study of resource allocation schemes, the distributed power
control algorithms for maximizing the network throughput have
attracted significant research attention \cite{FoschiniITVT1193,
YatesJSAC0995, SaraydarITC0202, HuangJSAC2006, EtkinTse2007,
JindalAllerton07}. Moreover, achieving a low transmission delay is
an important QoS requirement in wireless networks
\cite{Jung_Shah2007}. In particular, for buffer-limited users with
real-time services (e.g., interactive games, live sport videos,
etc), too much delay results in dropping some packets. Therefore,
the main challenge in wireless networks with real-time services is
to utilize an efficient power allocation scheme such that the delay
is minimized, while achieving a high throughput.

The throughput maximization problem in cellular and multihop
wireless networks has been extensively studied in
\cite{GuptaITIT2000, GrossglauserIACM0802, KulkaraniITIT0604,
XieITIT0504, HassibiITIT0706}. In these works, delay analysis is not
considered. However, it is shown that the high throughput is
achieved at the cost of a large delay \cite{ElgamalITIT0606}. This
problem has motivated the researchers to study the relation between
the delay characteristics and the throughput in wireless networks
\cite{LeAlfaITWC1106, BetteshPIMRC98, BansalINFOCOM2003,
ToumpisINFOCOM04}. In particular, in most recent literature
\cite{GopalaWNCMC2005, ElgamalITIT0606, SharmaITIT0606,
NeelyITIT0605, SahrifITWC0907, AlirezaITIT2007, ComaniciuITWC0806,
WangISIT2008, WalshISIT2008}, the tradeoffs between delay and
throughput have been investigated as a key measure of the network's
performance. The first studies on achieving a high throughput along
with a low-delay in ad hoc wireless networks are framed in
\cite{BansalINFOCOM2003} and \cite{ToumpisINFOCOM04}. This line of
work is further expanded in \cite{ElgamalITIT0606, SharmaITIT0606}
and \cite{NeelyITIT0605} by using different mobility models. El
Gamal \textit{et al.} \cite{ElgamalITIT0606} analyze the optimal
delay-throughput scaling for some wireless network topologies. For a
static random network with $n$ nodes, they prove that the optimal
tradeoff between throughput $T_{n}$ and delay $D_{n}$ is given by
$D_{n}=\Theta (nT_{n})$. Reference \cite{ElgamalITIT0606} also shows
that the same result is achieved in random mobile networks, when
$T_{n}=O(1/\sqrt{n\log n})$. Neely and Modiano \cite{NeelyITIT0605}
consider the delay-throughput tradeoff for mobile ad hoc networks
under the assumption of redundant packet transmission through
multiple paths. Sharif and Hassibi \cite{SahrifITWC0907} analyze the
delay characteristics and the throughput in a broadcast channel.
They propose an algorithm to reduce the delay without too much
degradation in the throughput. This line of work is further extended
in \cite{AlirezaITIT2007} by demonstrating that it is possible to
achieve the maximum throughput and short-term fairness
simultaneously in a large-scale broadcast network.

In \cite{JamshidITIT2008}, we addressed the throughput maximization
of a distributed single-hop wireless network with $K$ links, where
the links are partitioned into a fixed number ($M$) of clusters each
operating in a subchannel with bandwidth $\frac{W}{M}$. We proposed
a distributed and non-iterative power allocation strategy, where the
objective for each user is to maximize its best estimate (based on
its local information, i.e., direct channel gain) of the average
sum-rate of the network. Under the block Rayleigh fading channel
model with shadowing effect, it is established that the average
sum-rate in the network scales at most as $\Theta (\log K)$ in the
asymptotic case of $K \to \infty$. This order is achievable by the
distributed \textit{threshold-based on-off scheme} (i.e., links with
a direct channel gain above certain threshold $\tau_{n}$ transmit at
full power and the rest remain silent). In addition, in the strong
interference scenario, the on-off power allocation scheme is shown
to be the optimal strategy. Moreover, the optimum threshold level
that achieves the maximum average sum-rate of the network is
obtained as $\tau_{n}=\log n-2\log \log n +O(1)$, where
$n=\frac{K}{M}$ is the number of links in each cluster. We also
optimized the average network's throughput in terms of the number of
the clusters, $M$. It is proved that the maximum average sum-rate of
the network, assuming on-off power allocation scheme, is achieved at
$M=1$.  However, \cite{JamshidITIT2008} only focuses on the network
throughput and other issues (like delay and packet dropping
probability) were not addressed in this work.

In this paper, we follow the distributed single-hop wireless network
model proposed in \cite{JamshidITIT2008} with $M=1$ (which is the
case with the maximum throughput) and address the delay-throughput
tradeoff of the network. The channels are assumed to be
\textit{block Rayleigh fading with shadowing} (the same model as in
\cite{JamshidITIT2008}), where the transmission block is assumed to
be equal to the fading block (which is assumed to be equal for all
links). Moreover, the links are assumed to be synchronous. The
assumption of block Rayleigh fading with synchronous users is used
in many works in the literature (like \cite{marzetta} for the
point-to-point scenario, \cite{shlomo} for the multiple-access
channel,  and \cite{SahrifITWC0907} and \cite{AlirezaITIT2007} for
the broadcast scenario). We consider a buffer-limited network, in
which the users have a buffer size of one packet. This assumption
introduces \textit{dropping event} in the network, which is defined
as the event when a packet is arrived in the buffer while the
previous packet has not been served yet. Although the assumption of
one packet buffer size is harsh for many practical applications, it
simplifies the analysis while giving a good insight about the worst
case performance in the network. Noting the optimality of on-off
power allocation scheme in terms of achieving the maximum order of
the sum-rate throughput \cite{JamshidITIT2008}, we use it in this
work. Therefore, for any link, if the direct channel is above a
pre-determined threshold and there is any packet in the buffer, the
transmitter sends that packet during a transmission block with full
power and if not, remains silent.

  In the first part, we define a new notion of throughput,
called \textit{effective throughput}, which describes the
\textit{actual} amount of data transmitted through each links. This
notion captures the effect of arrival process by taking into account
the \textit{full buffer probability}. We compute the optimum
threshold level $\tau_n$, and the corresponding maximum effective
throughput of the network, for each packet arrival process. It is
proved that the effective throughput of the network scales as
$\frac{\log n}{\hat{\alpha}}$, with $\hat{\alpha} \triangleq \alpha
\varpi$, regarding the packet arrival process. This throughput
scaling is exactly the same as what we had derived in
\cite{JamshidITIT2008}, i.e., the case of backlogged users.
Moreover, we show that the maximum throughput is achieved in the
\textit{strong interference scenario}, in which the interference
term dominates the noise. As an interesting consequence, the results
of this section are valid even without the assumption of
synchronization between the users or equality of their fading
coherence time (fading blocks).

In the second part, we present the delay characteristics of the
underlying network in terms of the packet dropping probability for
deterministic and stochastic packet arrival processes. We derive the
sufficient conditions in the asymptotic case of $n \to \infty$ such
that the packet dropping probability of the links tends to zero,
while achieving the maximum effective throughput of the network,
asymptotically. The importance of this result is showing the fact
that the loss in the network performance due to the limited buffer
size can be made negligible in the asymptotic regime of $n \to
\infty$. In the subsequent section, we study the tradeoff between
the effective throughput of the network and other performance
measures, i.e., packet dropping probability and delay for different
arrival processes. In particular, we determine how much degradation
will be enforced in the throughput by introducing the aforementioned
constraints, and how much this degradation depends on the arrival
process. The setup in this paper is quite different from that of
with the on-off Bernoulli scheme in \cite{Neelyallerton06}. In fact,
we utilize a distributed approach using local information, i.e.,
direct channel gains, while \cite{Neelyallerton06} relies on a
central controller which studies the channel conditions of all the
links and decides accordingly. Furthermore, we consider a
homogeneous network model without path loss. This differs from the
geometric models considered in \cite{ElgamalITIT0606,
SharmaITIT0606} and \cite{NeelyITIT0605}, which are based on the
distance between the source and the destination (i.e., power
decay-versus-distance law).

The rest of the paper is organized as follows. In Section
\ref{model}, the network model and objectives are described. The
throughput maximization of the underlying network is presented in
Section \ref{analysis}. The delay characteristics in terms of the
packet dropping probability are analyzed in Section \ref{secdelay1}.
Section \ref{TDT} establishes the tradeoff between the throughput,
delay, and packet dropping probability in the underlying network.
Simulation results are presented in in section \ref{numerical_Ch4}.
Finally, in Section \ref{conclusion1}, an overview of the results
and conclusions are presented.

\textit{Notations:} For any functions $f(n)$ and $g(n)$
\cite{KnuthACM67}:
\begin{itemize}
\item [$\bullet$] $f(n)=O(g(n))$ means that $\lim_{n \to \infty} \Big \vert \frac{f(n)}{g(n)} \Big \vert < \infty$.
\item [$\bullet$] $f(n)=\mathit{o}(g(n))$ means that $\lim_{n \to \infty} \Big \vert \frac{f(n)}{g(n)} \Big \vert =0$.
\item [$\bullet$] $f(n)=\omega(g(n))$ means that $\lim_{n \to \infty} \frac{f(n)}{g(n)} = \infty$.
\item [$\bullet$] $f(n)=\Omega(g(n))$ means that $\lim_{n \to \infty}  \frac{f(n)}{g(n)}  > 0$.
\item [$\bullet$] $f(n)=\Theta (g(n))$ means that $\lim_{n \to \infty} \frac{f(n)}{g(n)}=c$, where $0<c<\infty$.
\item [$\bullet$] $f(n) \sim g(n)$ means that $\lim_{n \to \infty} \frac{f(n)}{g(n)}=1$.
\item [$\bullet$] $f(n) \approx g(n)$ means that $f(n)$ is approximately equal to $g(n)$, i.e., if
we replace $f(n)$ by $g(n)$ in the equations, the results still hold.
\end{itemize}

Throughout the paper, we use $\log(.)$ as the natural logarithm
function and $\mathbb{N}_{n}$ for representing the set $\lbrace 1,2,
\cdots, n \rbrace$. Also, $\mathbb{E}[.]$ represents the expectation
operator, and $\mathbb{P} \{.\}$ denotes the probability of the
given event.

\section{Network Model and Problem Description}\label{model}

\subsection{Network Model}

In this work, we consider a distributed single-hop wireless network,
in which $n$ pairs of nodes\footnote{The term ``pair" is used to
describe the transmitter and the related receiver, while the term
``user" is used only for the transmitter.}, indexed by
$\{1,...,n\}$, are located within the network area (Fig. \ref{fig:
model}). We assume the number of links, $n$, is known information
for the users. All the nodes in the network are assumed to have a
single antenna. Also, it is assumed that all the transmissions occur
over the same bandwidth. In addition, we assume that each receiver
knows its direct channel gain with the corresponding transmitter, as
well as the interference power imposed by other users. However, each
transmitter is assumed to be only aware of the direct channel gain
to its corresponding receiver. The power of Additive White Gaussian
Noise (AWGN) at each receiver is assumed to be $N_{0}$.
\begin{figure}[t]
\centerline{\psfig{figure=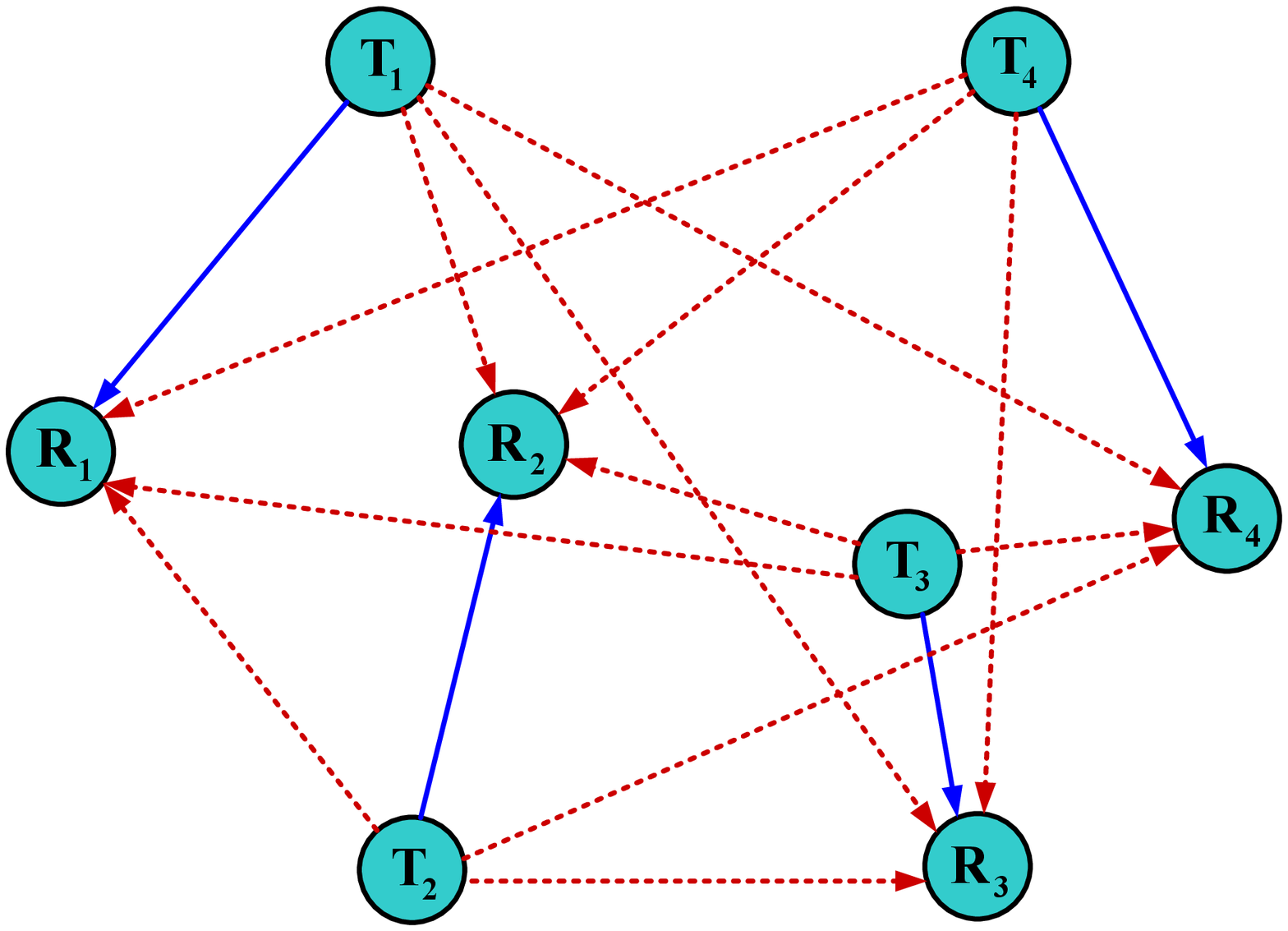,width=4.5in}}
\caption{A distributed single-hop wireless network with $n=4$.}
\label{fig: model}
\end{figure}

We assume that the time axis is divided into slots with the duration
of one transmission block, which is defined as the unit of time. The
channel model is assumed to be Rayleigh flat-fading with the
shadowing effect. The channel gain\footnote{In this paper,
\textit{channel gain} is defined as the square magnitude of the
\textit{channel coefficient}.} between transmitter $j$ and receiver
$i$ at time slot $t$ is represented by the random variable
$\mathcal{L}^{(t)}_{ji}$ \footnote{In the sequel, we use the
superscript $(t)$ for some events to show that the events occur in
time slot $t$.}. For $j=i$, the \textit{direct channel gain} is
defined as $\mathcal{L}^{(t)}_{ji} \triangleq h^{(t)}_{ii}$, where
$h^{(t)}_{ii}$ is exponentially distributed with unit mean (and unit
variance). For $j \neq i$, the \textit{cross channel gains} are
defined based on a shadowing model as follows\footnote{For more
details, the reader is referred to \cite{AbdiVTC1999} and
\cite{AbdiVTC2001} and references therein.}:
\begin{eqnarray}\label{eqn: crosschannel}
    \mathcal{L}^{(t)}_{ji} \triangleq
\left\{\begin{array}{ll}
\beta^{(t)}_{ji} h^{(t)}_{ji}   ,    & \textrm{with probability}~~~~\alpha \\
0 ,  &  \textrm{with probability}~~ 1-\alpha,
\end{array} \right.
\end{eqnarray}
where $h^{(t)}_{ji}$s have the same distribution as $h^{(t)}_{ii}$s,
$0 \leq \alpha \leq 1$ is a fixed parameter, and the random variable
$\beta^{(t)}_{ji}$, referred to as the \textit{shadowing factor}, is
independent of $h^{(t)}_{ji}$ and satisfies the following
conditions:
\begin{itemize}
\item $\beta_{min} \leq \beta^{(t)}_{ji} \leq \beta_{max}$, where $\beta_{min} >0$ and $\beta_{max}$ is finite,
\item $\mathbb{E} \big[\beta^{(t)}_{ji}\big] \triangleq \varpi \leq 1$.
\end{itemize}
All the channels in the network are assumed to be quasi-static block
fading, i.e., the channel gains remain constant during one block and
change independently from block to block. In other words,
$\mathcal{L}^{(t)}_{ji}$ is independent of $\mathcal{L}^{(t')}_{ji}$
for $t \neq t'$. Moreover, the fading block of all channels are
assumed to be equal to each other and this value is equal to the
duration of the transmission block for all users. This model is also
used in  \cite{SahrifITWC0907} and \cite{AlirezaITIT2007}. Also,
users are assumed to be synchronous to each other. However, as we
will see later, the results of the paper are still valid even in the
cases that the users are not synchronous or the fading block
(coherence time) of the channels are not equal.

\subsection{On-Off Power Allocation Strategy}\label{On-Off}
In \cite{JamshidITIT2008}, we have shown that a distributed scheme,
called \textit{threshold-based on-off scheme}, achieves the maximum
order of the  sum-rate throughput in a single-hop wireless network
with $n$ links,  under the block Rayleigh fading channel model
possibly with shadowing effect, in the asymptotic regime of $n \to
\infty$. Moreover, in the strong interference scenario, the on-off
power allocation scheme is the optimal strategy, in terms of the
sum-rate throughput, assuming the availability of direct channel
gains at the transmitters.

Motivated by the results of \cite{JamshidITIT2008}, we assume that
all the links utilize the threshold-based on-off power allocation
strategy proposed in \cite{JamshidITIT2008} \footnote{ We consider a
homogeneous network in the sense that all the links have the same
configuration and use the same protocol. Thus, the transmission
strategy for all users are agreed in advance.}. Unlike most of the
works in the literature that assume backlogged users, here we assume
a practical model for the packet arrivals in which the buffer of
each link is not necessarily full (of packet) all the time. Based on
this observation, we adopt the on-off power allocation scheme during
each time slot $t$ as follows:

1- Based on the direct channel gain, the transmission policy
is\footnote{In fact, if there is no packet in the buffer, it does
not make sense for the user to be active, even if its channel is
good. }
\begin{eqnarray}
    p^{(t)}_{i}=
\left\{\begin{array}{ll}
1   ,    & \textrm{if}~~ h^{(t)}_{ii}> \tau_{n}~\textrm{and the buffer of link}~i~\textrm{is full at time slot}~t \\
0 ,  &  \textrm{Otherwise},
\end{array} \right.
\end{eqnarray}
where $p^{(t)}_{i}$ denotes the transmission power of user $i$ at
time slot $t$ and $\tau_{n}$ is a pre-specified threshold level that
is a function of $n$ and also depends on the channel model and
packet arrival process.

2- Knowing its corresponding direct channel gain, each active user $i$
transmits a Gaussian signal with full power and the rate equal to:
\begin{equation}\label{eqn: 02}
R^{(t)}_{i}=\mathbb{E}_{h^{(t)}_{ii},I^{(t)}_i} \left[\log \left( 1+  \dfrac{h^{(t)}_{ii}p^{(t)}_{i}}{I^{(t)}_{i}+N_{0}} \right)\right]~~~\textrm{nats/channel use},
\end{equation}
where $I^{(t)}_{i}=\sum_{\substack{j=1 \\ j\neq
i}}^{n}\mathcal{L}^{(t)}_{ji}p^{(t)}_{j}$ is the power of the
interference term seen by receiver $i \in \mathbb{N}_{n}$ at time
slot $t$. The above rate is achievable by encoding and decoding over
arbitrarily large number ($M$) of blocks. More precisely, assuming
the number of channel uses per each transmission block to be $N$,
the $i^{th}$ transmitter maps the message $m \in \{m_1, m_2, \cdots,
m_L\}$, where $L=2^{MNR^{(t)}_{i}}$, to a Gaussian codeword of size
$MN$, $\mathcal{C}_m \in \{\mathcal{C}_1, \mathcal{C}_2, \cdots,
\mathcal{C}_L\}$. In the $k^{th}$ block, if $p^{(t)}_{i}=1$, the
transmitter sends the $k^{th}$ portion of $\mathcal{C}_m$, denoted
by $\mathcal{C}_m (k)$. At the receiver side, the decoder considers
only the blocks in which the transmitter was transmitting with full
power, denoted by $\{a_1, \cdots, a_l\}$, and is able to decode the
message $m$, if $ L \leq 2^{Nl R_1}$, where $R_1 \triangleq
\mathbb{E}_{h^{(t)}_{ii},I^{(t)}_i} \bigg[\log \bigg( 1+
\frac{h^{(t)}_{ii}}{I^{(t)}_{i}+N_{0}} \bigg)\bigg|
p^{(t)}_{i}=1\bigg]$. Noting that as $M \to \infty$, $l \approx M
\mathbb{P} \{p^{(t)}_{i}=1\}$, and $R^{(t)}_{i} = \mathbb{P}
\{p^{(t)}_{i}=1\} R_1$, it is concluded that the rate $R^{(t)}_{i}$
is achievable. As we will see later, in the optimal performance
regime, which is the strong interference regime, encoding and
decoding over single blocks is sufficient to achieve (\ref{eqn:
02}).

\subsection{Packet Arrival Process}
One of the most important parameters in the network analysis is the
model for the packet arrival process. The packet arrival process is
a random process which is described by either the arrival time of
the packets or the interarrival time between the subsequent packets.
These quantities may be modeled by the deterministic or stochastic
processes (Fig. \ref{fig: Delay1}). In this paper, we consider the
following packet arrival processes:

\begin{itemize}
\item \textit{Poisson Arrival Process (PAP):} In this process, the number of arrived packets in any interval of unit length  is assumed
to have a Poisson distribution with the parameter
$\frac{1}{\lambda}$. This process is a commonly used model for
random and mutually independent packet arrivals in queueing theory
\cite{Gallagerbook95}.
\item \textit{Bernoulli Arrival Process (BAP):} In this process, at any given time slot, the probability that
a packet arrives is $\rho \triangleq \frac{1}{\lambda}$ \footnote{We
choose the parameter $\rho$ as $\frac{1}{\lambda}$ to be consistent
with other packet arrival processes.}. Moreover, the arrival of the
packets in different slots occurs independently. This model has been
used in many works in the literature such as \cite{NeelyITIT0605}
and \cite{ZorziAllerton98}.
\item \textit{Constant Arrival Process (CAP):} In this process, packets arrive continuously with a constant rate of $\frac{1}{\lambda}$ packets per
unit length (Fig. \ref{fig: Delay1}-b) \cite{TangITWC1207}.
\end{itemize}
\begin{figure}[bhpt]
\centerline{\psfig{figure=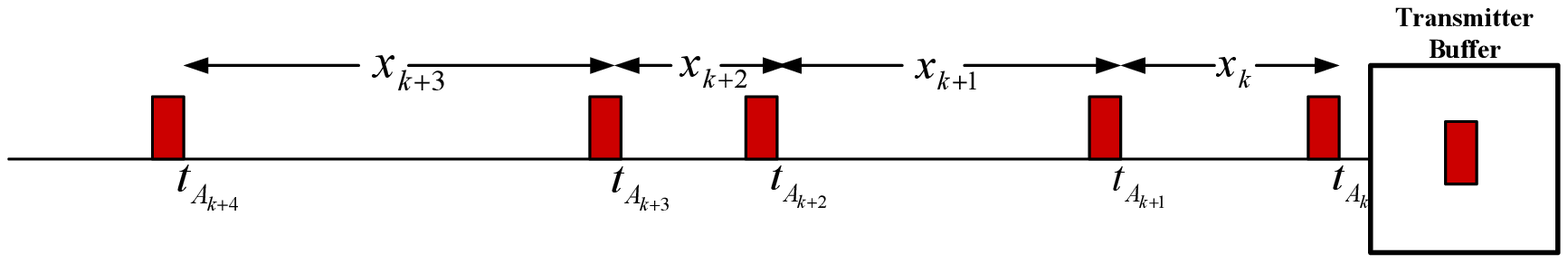,width=6.6in}}
\vspace{-7pt} \center{\hspace{16pt} \small{(a)}} \vspace{10pt}
\hspace{1pt} \centerline{\psfig{figure=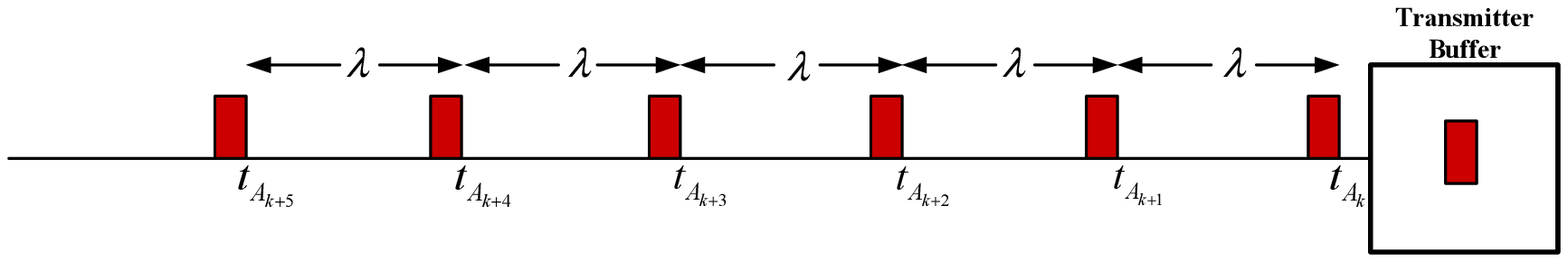,width=6.6in}}
\vspace{-35pt}
\center{\hspace{14pt} \small{(b)}} \\
\vspace{-7pt} \caption[a) .] {\small{A schematic figure for a) stochastic packet arrival process, b) constant packet arrival process.}}
\label{fig: Delay1}
\end{figure}
It is assumed that the packet arrival process for all links is the
same. Let us denote $t_{A_{k}}^{(i)}$ as the time instant of the
$k^{th}$ packet arrival into the buffer of link $i$. It is observed
from Fig. \ref{fig: Delay1}-a that
$t_{A_{k}}^{(i)}=\sum_{j=1}^{k-1}x^{(i)}_{j}+t^{(i)}_{0}$ where
$t^{(i)}_{0}$ is the starting time for link $i$, and the random
variable $x^{(i)}_{j}$ is the interarrival time defined as
\begin{equation}\label{inter01}
x^{(i)}_{j} \triangleq t^{(i)}_{A_{j+1}}-t^{(i)}_{A_{j}},
\end{equation}
with $\mathbb{E}[x^{(i)}_{j}]=\lambda$. For the CAP,
$x^{(i)}_{j}=\lambda$ and
$t_{A_{k}}^{(i)}=(k-1)\lambda+t^{(i)}_{0}$\footnote{For analysis
simplicity, we assume that $\lambda$ is an integer number.}, while
for the PAP, $x^{(i)}_{j}$'s are independent samples of an
exponential random variable $x$ with the probability density
function (pdf)
\begin{equation}\label{pdf1}
 f_{X}(x)=\dfrac{1}{\lambda}e^{-\frac{1}{\lambda}x},~~~~ x>0.
\end{equation}
Also for the BAP, $x^{(i)}_{j}$'s are independent samples of a
geometric random variable $X$ with the probability mass function
(pmf)
\begin{equation}\label{pmf1}
 p_{X}(m) \triangleq \mathbb{P} \{X=m\}=(1-\rho)^{m-1}\rho,~~~~ m=1,2,...,
\end{equation}
with $\rho \triangleq \frac{1}{\lambda}$.

We represent $t_{D_{k}}^{(i)}$ as the time instant at which either
the $k^{th}$ arriving packet departs the buffer of link $i$ for the
transmission or drops from the buffer. In such configuration, we
have the following definition:

\begin{defi}
\textbf{(Delay):} The random variable $\mathscr{D}^{(i)}_{k}
\triangleq t_{D_{k}}^{(i)}-t_{A_{k}}^{(i)}$ for each link $i$ is
defined as the delay between the departure and the arrival time of
each packet $k$, expressed in terms of the number of time slots.
\end{defi}

In this work, we assume that the buffer size for each transmitter is
one packet. Due to the this limitation on the buffer size and the
on-off power allocation strategy, the existing buffered packet may
be dropped if it is not served before the arrival of the next
packet. Mathematically speaking, the event that the dropping of
packet $k$ occurs in link $i\in \mathbb{N}_{n}$ is defined as
\begin{eqnarray}
\label{eqn: event01}\mathscr{B}_{i} & \equiv & \left \lbrace \mathscr{D}^{(i)}_{k} \geq t_{A_{k+1}}^{(i)}-t_{A_{k}}^{(i)} \right \rbrace \\
\label{eqn: event1} &\equiv& \left \lbrace \mathscr{D}^{(i)}_{k} \geq x^{(i)}_{k} \right \rbrace.
\end{eqnarray}
Therefore, the packet dropping probability in each link $i\in
\mathbb{N}_{n}$, denoted by $\mathbb{P} \left \{\mathscr{B}_{i}
\right \}$, can be obtained as
\begin{eqnarray}
\label{eqn: 03} \mathbb{P} \left \{\mathscr{B}_{i} \right \} &=& \mathbb{P} \left \lbrace \mathscr{D}^{(i)}_{k} \geq x^{(i)}_{k} \right \rbrace \\
&=& \int_{0}^{\infty}\mathbb{P} \left \lbrace \mathscr{D}^{(i)}_{k} \geq x^{(i)}_{k} \Big| x^{(i)}_{k}=x \right \rbrace f_{X}(x)dx,~~~~~\textrm{for PAP},\\
\label{eqn: BAP1}&=& \sum_{m=1}^{\infty}\mathbb{P} \left \lbrace \mathscr{D}^{(i)}_{k} \geq x^{(i)}_{k} \Big| x^{(i)}_{k}=m  \right \rbrace p_{X}(m),~~~~~~~\textrm{for BAP}, \\
\label{eqn: CAP} &=& \mathbb{P} \left \lbrace \mathscr{D}^{(i)}_{k}
\geq \lambda \right \rbrace,
~~~~~~~~~~~~~~~~~~~~~~~~~~~~~~~~~~~\textrm{for CAP}.
\end{eqnarray}
where $f_{X}(x)$ and $p_{X}(m)$ are defined as (\ref{pdf1}) and
(\ref{pmf1}), respectively. In Section \ref{secdelay1}, we will
obtain $\mathbb{P} \left \{\mathscr{B}_{i} \right \}$ for different
packet arrival processes in terms of $\lambda$ and $\tau_n$.

\subsection{Objectives}

\textbf{Part I: Throughput Maximization:} The main objective of the first
part of this paper is to maximize the throughput of the underlying network.
To address this problem, we first define a new notion of throughput,
called \textit{effective throughput}, which denotes the \textit{actual}
amount of data transmitted through the links. In order to derive the
effective throughput, we obtain the \textit{full buffer probability}
of a link for the deterministic and stochastic
packet arrival processes. Then, we compute the optimum threshold level
$\tau_n$, and the maximum effective throughput of the network, for each
packet arrival process.

\textbf{Part II: Delay Characteristics:} The main objective of the
second part is to formulate the packet dropping probability of each
link in the underlying network based on the aforementioned packet
arrival processes  in terms of the number of links ($n$), $\lambda$,
and the parameter of the on-off power allocation scheme ($\tau_n$).
This analysis enables us to derive the sufficient conditions in the
asymptotic case of $n \to \infty$ such that the packet dropping
probabilities tend to zero, while achieving the maximum effective
throughput of the network.

\textbf{Part III: Delay-Throughput-Dropping Probability Tradeoff:}
The main goal of the third part is to study the tradeoff between the
effective throughput of the network and other performance measures,
i.e., the dropping probability and the delay-bound ($\lambda$) for
different packet arrival processes. In particular, we are interested
to determine how much degradation will be enforced in the throughput
by introducing the other constraints, and how much this degradation
depends on the packet arrival process.

\section{Throughput Maximization}\label{analysis}

\subsection{Effective Throughput}
In this section, we aim to derive the maximum throughput of the
network with a large number ($n$) of links, based on using the
distributed on-off power allocation strategy. We present a new
performance metric in the network, called \emph{effective
throughput}, which is a function of the threshold level $\tau_{n}$
and $\lambda$. Let us start with the following definition.

\begin{defi}
\textbf{(Effective Throughput):} Under the on-off power allocation
strategy, the effective throughput of each link $i$, $i \in
\mathbb{N}_{n}$, is defined (on a per-block basis) as
\begin{equation}\label{thr}
\mathfrak{T}_{i} \triangleq \lim_{L \to \infty}\dfrac{1}{L}\sum_{t=1}^{L}R_{i}^{(t)}\mathcal{I}_{i}^{(t)},
\end{equation}
where $R_{i}^{(t)}$ is defined as (\ref{eqn: 02}) and
$\mathcal{I}_{i}^{(t)}$ is an indicator variable which is equal to
$1$, if user $i$ transmits at time slot $t$, and $0$ otherwise.
Furthermore, the effective throughput of the network is defined as
\begin{eqnarray}
 \mathfrak{T}_{\mathrm{eff}} \triangleq \sum_{i=1}^n \mathfrak{T}_i.
\end{eqnarray}

\end{defi}

The quantity $\mathfrak{T}_{i}$ represents the average amount of
information conveyed through link $i$ in a long period of time. This
metric is suitable for real-time applications, where the packets
have a certain amount of information and certain arrival rates. It
should be noted that $\mathcal{I}_{i}^{(t)}=1$ is equivalent to the
case in which the buffer is full and the channel gain $h^{(t)}_{ii}$
is greater than the threshold level $\tau_{n}$ at time slot $t$.
Defining the full buffer event as follows
\begin{equation}
\mathscr{C}_{i}^{(t)} \equiv \{\textrm{Buffer of link}~i~\textrm{is full at time slot}~t \},
\end{equation}
 we have
\begin{eqnarray}
\label{index01}\mathbb{P} \left \lbrace \mathcal{I}_{i}^{(t)}=1 \right \rbrace & =& \mathbb{P}\left \lbrace h^{(t)}_{ii}>\tau_{n},~ \mathscr{C}_{i}^{(t)}\right \rbrace\\
&\stackrel{(a)}{=}&\mathbb{P}\left \lbrace h^{(t)}_{ii}>\tau_{n}\right \rbrace \mathbb{P}\left \lbrace \mathscr{C}_{i}^{(t)} \right \rbrace\\
\label{index02}&=& q_{n}\Delta_{n},
\end{eqnarray}
where $q_{n} \triangleq \mathbb{P} \left\lbrace h^{(t)}_{ii} >
\tau_{n}\right\rbrace$, and $\Delta_{n} \triangleq \mathbb{P} \left
\lbrace \mathscr{C}_{i}^{(t)} \right \rbrace$ is the \textit{full
buffer probability}. In the above equations, $(a)$ follows from the
fact that the full buffer event depends on the packet arrival
process as well as the direct channel gains $h^{(t^{'})}_{ii}$, for
$t^{'} < t$, which is independent of the channel gain $h^{(t)}_{ii}$
(due to the block fading channel model). Thus,
\begin{eqnarray}\label{qdn}
    \mathcal{I}_{i}^{(t)}=
\left\{\begin{array}{ll}
1   ,    &  \textrm{with  probability} \, \, q_{n}\Delta_{n},  \\
0   ,    &  \textrm{with  probability} \, \, 1-q_{n}\Delta_{n}.
\end{array} \right.
\end{eqnarray}
It is observed that $\mathcal{I}_{i}^{(t)}$ is a Bernoulli random
variable with parameter $q_{n}\Delta_{n}$. In fact,
$q_{n}\Delta_{n}$ is the probability of the link activation which is
a function of $n$. In the sequel, we derive $\Delta_n$ for the
aforementioned packet arrival processes.

\subsection{Full Buffer Probability}
Let us denote $t^{(i)}_{a}$ as the time instant the last packet has
arrived in the  buffer of link $i$ before or at the same time $t$.
The event $\mathscr{C}_{i}^{(t)}$ implicitly indicates that during
$\mathscr{X}_i^{(t)} \triangleq t-t^{(i)}_{a}$ time slots, the
channel gain of link $i$ is less  than the threshold level
$\tau_{n}$. Clearly, $\mathscr{X}_i^{(t)}$ is a random variable
which varies from zero to infinity for the stochastic packet arrival
processes and is finite for the CAP\footnote{Note that, here we
assume that if a packet arrives at time $t$ and the channel gain is
greater than $\tau_n$ at this time, the packet will be
transmitted.}. Under the on-off power allocation scheme and using
the block fading  model property, the full buffer probability can be
obtained as\footnote{As we will show in Lemma \ref{lem004},
$\Delta_{n}$ is independent of index $i$.}
\begin{equation}\label{buffer01}
\Delta_{n}=\mathbb{E}\left[ (1-q_{n})^{\mathscr{X}_i^{(t)}} \right],
\end{equation}
where the expectation is computed with respect to
$\mathscr{X}_i^{(t)}$, and $q_{n} \triangleq \mathbb{P} \left\lbrace
h^{(t)}_{ii} > \tau_{n}\right\rbrace=e^{-\tau_n}$.

\begin{lem} \label{lem004}
Let us denote the full buffer probability of an arbitrary link $i
\in \mathbb{N}_{n}$, for the Poisson, Bernoulli and constant arrival
processes as $\Delta^{PAP}_{n}$, $\Delta^{BAP}_{n}$ and
$\Delta^{CAP}_{n}$, respectively. Then,
\begin{eqnarray}
\label{PAP001}\Delta^{PAP}_{n} &=&\dfrac{1}{1+\lambda \log (1-q_{n})^{-1}}, \\
\label{BAP001}\Delta^{BAP}_{n} &=&\dfrac{1}{1+(\lambda-1)q_{n}},\\
\label{CAP001}\Delta^{CAP}_{n} &=&\dfrac{1-(1-q_{n})^{\lambda}}{\lambda q_{n}}.
\end{eqnarray}
\end{lem}

\begin{proof}
For the PAP, since $\mathscr{X}_i^{(t)}$ is an exponential random
variable, (\ref{buffer01}) can be simplified as
\begin{eqnarray}
\Delta^{PAP}_{n} &=& \int_{0}^{\infty}\dfrac{1}{\lambda}(1-q_{n})^{x}e^{-\frac{1}{\lambda}x} dx\\
\label{buffer02}&=& \dfrac{1}{1+\lambda \log (1-q_{n})^{-1}}.
\end{eqnarray}

Also for the BAP, $\mathscr{X}_i^{(t)}$ is a geometric random
variable with parameter $\rho=\frac{1}{\lambda}$. Thus,
(\ref{buffer01}) can be simplified as
\begin{eqnarray}
\Delta^{BAP}_{n} &=& \sum_{m=0}^{\infty}(1-q_{n})^{m}\rho(1-\rho)^{m}\\
\label{buffer03}&\stackrel{(a)}{=} & \dfrac{1}{1+(\lambda-1)q_{n}},
\end{eqnarray}
where $(a)$ follows from the following geometric series:
\begin{equation}
\sum_{m=0}^{\infty}x^m=\dfrac{1}{1-x},~~|x|<1.
\end{equation}

For the CAP, the full buffer probability in (\ref{buffer01}) can be
written as
\begin{eqnarray}
\Delta^{CAP}_{n}&\stackrel{(a)}{=}&\sum_{m=0}^{\lambda-1} (1-q_{n})^m \mathbb{P} \{\mathscr{X}_i^{(t)}=m \}\\
&\stackrel{(b)}{=}&\sum_{m=0}^{\lambda-1} (1-q_{n})^m \dfrac{1}{\lambda}\\
&\stackrel{(c)}{=}&\dfrac{1-(1-q_{n})^{\lambda}}{\lambda q_{n}},
\end{eqnarray}
where $(a)$ follows from Fig. \ref{fig: Delay1}-b, in which
$\mathscr{X}_i^{(t)}$ varies from zero to $\lambda-1$ and $(b)$
follows from the fact that for the deterministic process,
$\mathscr{X}_i^{(t)}$ has a uniform distribution. In other words,
for every value of $m \in [0,\lambda-1]$, $\mathbb{P}
\{\mathscr{X}_i^{(t)}=m \}=\frac{1}{\lambda}$. Also, $(c)$ comes
from the following geometric series:
\begin{equation}\label{geo1}
\sum_{m=0}^{s}x^m=\dfrac{1-x^{s+1}}{1-x}.
\end{equation}
\end{proof}
Having derived the full buffer probability, we obtain the effective
throughput of the network in the following section.

\subsection{Effective Throughput of the Network}

Rewriting (\ref{thr}), the effective throughput of link $i$ can be
obtained as
\begin{eqnarray}
\mathfrak{T}_i &=& \lim_{L \to \infty}\dfrac{1}{L}\sum_{t=1}^{L}R_{i}^{(t)}\mathcal{I}_{i}^{(t)}\\
&\stackrel{(a)}{=}& \mathbb{E} \left[R_{i}^{(t)} \mathcal{I}_i^{(t)}\right]\\
&=& \mathbb{E} \left. \left[ R_{i}^{(t)} \mathcal{I}_i^{(t)}\right| \mathcal{I}_i^{(t)}=1\right]\mathbb{P} \left \lbrace \mathcal{I}_{i}^{(t)}=1 \right \rbrace+ \mathbb{E} \left. \left[ R_{i}^{(t)} \mathcal{I}_i^{ (t)} \right| \mathcal{I}_i^{(t)}=0\right]\mathbb{P} \left \lbrace \mathcal{I}_{i}^{(t)}=0 \right \rbrace  \label{efft}\\
&\stackrel{(b)}{=}& q_n \Delta_{n} \mathbb{E} \left. \left[ R_{i}^{(t)} \right| h_{ii}^{(t)}>\tau_n , \mathscr{C}_i^{(t)}\right] \\
\label{effect01}&\stackrel{(c)}{=}& q_n \Delta_{n} \mathbb{E}  \left. \left[\log \left( 1+\frac{h_{ii}^{(t)}}{I_i ^{(t)}+N_0}\right)\right| h_{ii}^{(t)}>\tau_n\right] \label{Ti},
\end{eqnarray}
where the expectation is computed with respect to $h^{(t)}_{ii}$ and
the interference term $I^{(t)}_{i}$. In the above equations, $(a)$
follows from the ergodicity of the channels (due to the block fading
model), which implies that the average over time is equal to average
over realization.  $(b)$ results from
(\ref{index01})-(\ref{index02}) and $\mathbb{E} \big[ R_{i}^{(t)}
\mathcal{I}_i^{(t)} \big| \mathcal{I}_i^{(t)}=0\big] =0$. Finally,
$(c)$ results from the fact that the term $\log \left(
1+\frac{h_{ii}^{(t)}}{I_i ^{(t)}+N_0}\right)$ is independent of
$\mathscr{C}_i^{(t)}$.

In order to derive the effective throughput, we need to find the
statistical behavior of $I^{(t)}_i$ which is performed in the
following lemmas:

\begin{lem} \label{lemma000_Ch5}
Under the on-off power scheme, we have
\begin{equation}
\mathbb{E}\left[ I^{(t)}_i \right]=(n-1)\hat{\alpha} q_{n}\Delta_{n},
\end{equation}
\begin{equation}
\textrm{Var} \left[I^{(t)}_i \right] \leq (n-1)(2\alpha \kappa q_{n}\Delta_{n}),
\end{equation}
where $\hat{\alpha} \triangleq \alpha \varpi$ and $\kappa \triangleq \mathbb{E}\left[\left(\beta^{(t)}_{ji}\right)^{2} \right] $.
\end{lem}
\begin{proof}
See Appendix \ref{append01_a}.
\end{proof}

\begin{lem} \label{lemma000}
The maximum effective throughput is achieved at $\lambda = o (n)$
and the strong interference regime which is defined as $\mathbb{E}
[I^{(t)}_i]=\omega (1)$, $i \in \mathbb{N}_n$.
\end{lem}

\begin{proof}
Suppose that $\lambda \neq o (n)$ which implies that $\lambda =
\Omega (n)$. Using (\ref{Ti}), we have
\begin{eqnarray}
\mathfrak{T}_i &\leq& q_n \Delta_n \mathbb{E}  \left. \left[\log \left( 1+\frac{h_{ii}^{(t)}}{N_0}\right)\right| h_{ii}^{(t)}>\tau_n\right] \\
&\stackrel{(a)}{\leq}& q_n \Delta_n \log \left( 1+ \frac{\mathbb{E} \left. \left[ h_{ii}^{(t)} \right| h_{ii}^{(t)}>\tau_n\right]}{N_0} \right) \\
&=& q_n \Delta_n \log \left( 1+ \frac{\tau_n +1}{N_0} \right), \label{Ti1}
\end{eqnarray}
where $(a)$ comes from the concavity of $\log (.)$ function and
\textit{Jensen's inequality}, $\mathbb{E}\left[\log x
\right]\leq\log (\mathbb{E}\left[x \right])$, $x>0$. Following
(\ref{PAP001}) - (\ref{CAP001}), it is revealed that $\Delta_n \leq
\min \left(1,\frac{1}{\lambda q_n} \right)$ for all packet arrival
processes. Substituting in (\ref{Ti1}), we have
\begin{eqnarray}
\mathfrak{T}_i &\leq& \frac{1}{\lambda} \log \left(1+ \frac{\log \lambda +1}{N_0} \right) \notag\\
&\sim& \frac{\log \log \lambda}{\lambda},
\end{eqnarray}
which follows from the fact that the maximum value of $q_n \Delta_n
\log \left( 1+ \frac{\tau_n +1}{N_0} \right)$ with the condition of
$\Delta_n \leq \min \left(1,\frac{1}{\lambda q_n} \right)$ is
attained at $q_n = \frac{1}{\lambda}$. Noting that $\lambda = \Omega
(n)$, we have $\mathfrak{T}_i \leq \Theta \left(\frac{\log \log
n}{n}\right)$.

Now, suppose that $\lambda = o(n)$ but $\mathbb{E} [I^{(t)}_i] \neq
\omega (1)$, or equivalently, $\mathbb{E} [I^{(t)}_i] = O (1)$ for
some $i$. Since $\mathbb{E} [I^{(t)}_i] = (n-1) \hat{\alpha} q_n
\Delta_n$, the condition $\mathbb{E} [I^{(t)}_i] = O (1)$ implies
that there exists a constant $c$ such that $q_n \Delta_n \leq
\frac{c}{n}$. Noting (\ref{PAP001}) - (\ref{CAP001}), it follows
that either $\Delta_n \sim \frac{1}{\lambda q_n}$ or $\Delta_n =
\Theta (1)$. In the first case, the condition $q_n \Delta_n \leq
\frac{c}{n}$ implies that $n \leq c \lambda$ which cannot hold due
to the assumption of $\lambda = o(n)$. Therefore, we must have $q_n
\leq \frac{c'}{n}$, for some constant $c'$. Substituting in
(\ref{Ti1}) yields
\begin{eqnarray}
 \mathfrak{T}_i &\leq& \frac{c'}{n} \log \left( 1+\frac{\tau_n +1}{N_0} \right) \notag\\
&\stackrel{(a)}{\leq}& \frac{c'}{n} \log \left( 1+\frac{\log (n/c') +1}{N_0} \right) \notag\\
&=& \Theta \left( \frac{\log \log n}{n}\right),
\end{eqnarray}
where $(a)$ results from the fact that $q_n \log \left(
1+\frac{\tau_n +1}{N_0} \right)$ is an increasing function of $q_n$
and reaches its maximum at the boundary which is $\frac{c'}{n}$.

In the sequel, we present a lower-bound on the effective throughput
of link $i$ in the region $\lambda = o(n)$ and $\mathbb{E}
[I^{(t)}_i] = \omega (1)$ and show that this lower-bound beats the
upper-bounds derived in the other regions, proving the desired
result. For this purpose, using (\ref{Ti}), we write
\begin{eqnarray}
 \mathfrak{T}_i &\stackrel{(a)}{\geq}& q_n \Delta_{n} \log \left( 1+\frac{\tau_n}{\mathbb{E} \left. \left[ I^{(t)}_i \right| h^{(t)}_{ii} > \tau_n \right] + N_o} \right) \notag\\
&\stackrel{(b)}{=}& q_n \Delta_{n} \log \left( 1+\frac{\tau_n}{ (n-1) \hat{\alpha} q_n \Delta_n + N_o} \right) \notag\\
&\stackrel{(c)}{\approx}& q_n \Delta_{n} \log \left( 1+\frac{\tau_n}{ (n-1) \hat{\alpha} q_n \Delta_n } \right),
\end{eqnarray}
where $(a)$ follows from the convexity of the function $\log
(1+\frac{b}{x+a})$ with respect to $x$ and Jensen's inequality,
$(b)$ results from the independence of $I^{(t)}_i$ from
$h^{(t)}_{ii}$, and $(c)$ follows from neglecting the term $N_0$
with respect to $(n-1) \hat{\alpha} q_n \Delta_n$ due to the strong
interference assumption. Setting $q_n = \frac{\log^2 n}{n}$ and
$\lambda = \frac{n}{\log^2 n}$, it is easy to check that
$\frac{\tau_n}{ (n-1) \hat{\alpha} q_n \Delta_n } = o(1)$ and hence,
$\log \left( 1+\frac{\tau_n}{ (n-1) \hat{\alpha} q_n \Delta_n }
\right) \approx \frac{\tau_n}{ (n-1) \hat{\alpha} q_n \Delta_n }$
which gives the effective throughput as $\frac{\tau_n}{(n-1)
\hat{\alpha}} = \Theta \left( \frac{\log n}{n} \right)$ which is
greater than the throughput obtained in the other regimes.
\end{proof}
Due to the result of Lemma \ref{lemma000}, we restrict ourselves to
the case of $\lambda = o(n)$ and the strong interference regime in
the rest of the paper.

\begin{lem} \label{lemma001} Let us assume $0<\alpha \leq 1$ is
fixed and we are in the strong interference regime (i.e., $\mathbb{E}\left[I^{(t)}_i \right]= \omega (1)$).
Then \textit{with probability one} (w. p. 1),
we have
\begin{equation}
 I^{(t)}_{i}~\sim~ (n-1) \hat{\alpha} q_{n}\Delta_{n},
\end{equation}
as $n \to \infty$. More precisely, substituting $I^{(t)}_i$ by
$(n-1)\hat{\alpha} q_{n}\Delta_{n}$ does not change the asymptotic
effective throughput of the network.
\end{lem}
\begin{proof}
See Appendix \ref{append01}.
\end{proof}

\begin{lem} \label{lemma005}
The effective throughput of the network for large values of $n$ can
be obtained as
\begin{equation}\label{opio}
\mathfrak{T}_{\mathrm{eff}} \approx n q_n \Delta_{n} \log \left( 1+\frac{\tau_n}{n \hat{\alpha} q_n \Delta_n}\right).
\end{equation}
\end{lem}
\begin{proof}
Using (\ref{effect01}), the effective throughput of the network in
the asymptotic case of $n \to \infty$ is obtained as
\begin{eqnarray}
\mathfrak{T}_{\mathrm{eff}} &=& \sum_{i=1}^{n}\mathfrak{T}_i\\
&\stackrel{(a)}{\approx}& n q_n \Delta_{n} \mathbb{E}  \left. \left[\log \left( 1+\frac{h_{ii}^{(t)}}{(n-1) \hat{\alpha} q_{n}\Delta_{n}+N_0}\right)\right| h_{ii}^{(t)}>\tau_n\right]\\
\label {eqn:q10}&\stackrel{(b)}{\approx}&n q_n \Delta_{n} \mathbb{E}  \left. \left[\log \left( 1+\frac{h_{ii}^{(t)}}{n\hat{\alpha} q_{n}\Delta_{n}}\right)\right| h_{ii}^{(t)}>\tau_n\right],
\end{eqnarray}
where $(a)$ results from the strong interference assumption and
Lemma \ref{lemma001}, and $(b)$ follows from approximating $(n-1)
\hat{\alpha} q_{n}\Delta_{n}+N_0$  by $n\hat{\alpha}
q_{n}\Delta_{n}$ due to the strong interference assumption and large
values of $n$. A lower-bound on (\ref{eqn:q10}) can be written as
\begin{eqnarray} \label{kosekhar1}
 \mathfrak{T}_{\mathrm{eff}}^l = n q_n \Delta_{n} \log \left( 1+\frac{\tau_n}{n\hat{\alpha} q_{n}\Delta_{n}}\right).
\end{eqnarray}
Furthermore, due to the concavity of $\log (.)$ function and
Jensen's inequality, an upper-bound on $\mathfrak{T}_{\mathrm{eff}}$
can be given as
\begin{eqnarray} \label{kosekhar2}
 \mathfrak{T}_{\mathrm{eff}}^u &=& n q_n \Delta_{n} \log \left( 1+\frac{\mathbb{E} \left. \left[ h_{ii}^{(t)} \right| h_{ii}^{(t)} > \tau_n\right]}{n\hat{\alpha} q_{n}\Delta_{n}}\right) \notag\\
&=& n q_n \Delta_{n} \log \left( 1+\frac{\tau_n +1}{n\hat{\alpha} q_{n}\Delta_{n}}\right).
\end{eqnarray}
In order to prove that the above upper and lower bounds have the
same scaling, it is sufficient to show that the optimum threshold
value ($\tau_n$) is much larger than one. For this purpose, we note
that if $\tau_n = O(1)$, then the effective throughput of the
network will be upper-bounded by
\begin{eqnarray}
\mathfrak{T}_{\mathrm{eff}} &\stackrel{(a)}{\leq}& \dfrac{\tau_{n}+1}{\hat{\alpha}}\\
&=& O(1),
\end{eqnarray}
where $(a)$ follows from $\log (1+x) \leq x$. In other words, the
effective throughput of the network does not scale with $n$, while
the throughput of $\Theta (\log n)$, as will be shown later, is
achievable. This suggests that the optimum threshold value must grow
with $n$, and hence, the bounds given in  (\ref{kosekhar1}) and
(\ref{kosekhar2}) are asymptotically equal to $n q_n \Delta_{n} \log
\left( 1+\frac{\tau_n}{n\hat{\alpha} q_{n}\Delta_{n}}\right)$ and
this completes the proof of the lemma.
\end{proof}

\begin{lem} \label{lemma006}
The maximum effective throughput of the network is obtained in  the
region that $\tau_n =o \left( n \hat{\alpha}q_n \Delta_n \right)$.
\end{lem}
\begin{proof}
 Rewriting the expression of the effective throughput of the network from (\ref{opio}) and noting the fact that $\log (1+x) \leq x$, for $x \geq 0$, we have
\begin{eqnarray}
 \mathfrak{T}_{\mathrm{eff}} &\approx&  n q_n \Delta_{n} \log \left( 1+\frac{\tau_n}{n\hat{\alpha}
 q_{n}\Delta_{n}}\right) \notag\\
&\leq& \frac{\tau_n}{\hat{\alpha}}.
\end{eqnarray}
It can be shown that if the condition $\tau_n =o \left( n
\hat{\alpha}q_n \Delta_n \right)$ is not satisfied, the ratio
$\frac{\log \left( 1+\frac{\tau_n}{n\hat{\alpha}
 q_{n}\Delta_{n}}\right)}{\frac{\tau_n}{n\hat{\alpha}
 q_{n}\Delta_{n}}}$ is strictly less than one. Having $\tau_n =o \left( n \hat{\alpha}q_n \Delta_n \right)$ results in $\log \left( 1+\frac{\tau_n}{n\hat{\alpha}
 q_{n}\Delta_{n}}\right) \approx \frac{\tau_n}{n\hat{\alpha}
 q_{n}\Delta_{n}}$ yielding the upper-bound $\frac{\tau_n}{\hat{\alpha}}$. This means that to achieve the
maximum throughput, the interference should not only be strong but also be much larger than $\tau_n$.
\end{proof}

\textit{Observation -} An interesting observation of Lemmas
\ref{lemma000}-\ref{lemma006} is that there is no need to have
synchronization between the users or equality of the fading blocks
(coherence time) of the channels to obtain these results. This is
due to the fact that during a transmission block (which is equal to
the fading block of the corresponding direct channel), the receiver
observes different samples of interference $I_i$ (due to
asynchronousy between the users). However, as the interference is
strong, from the result of Lemma \ref{lemma001}, all samples of
interference  asymptotically almost surely scale as $n \hat{\alpha}
q_n \Delta_n$, and hence, the receiver is still capable of decoding
the message correctly if the transmission rate is below $q_n
\Delta_n \log \left(1+\frac{\tau_n}{n \hat{\alpha} q_n \Delta_n}
\right)$. Moreover, the encoding and decoding do not need to be
performed over large number of blocks. In fact, in the blocks where
$h_{ii}^{(t)}>\tau_n$, the transmitter sends data with the rate
$\log \left(1+\frac{\tau_n}{n \hat{\alpha} q_n \Delta_n} \right)$
nats/channel use and the decoder will be able to decode the packet
information correctly.

Having the expression for the effective throughput of the network in
(\ref{opio}), in the next theorem, we find the optimum value of
$q_{n}$ (or equivalently $\tau_n$) in terms of $n$ and $\lambda$ for
the aforementioned packet arrival processes, i.e.:
\begin{equation}\label{opti01}
\hat{q}_{n}= \arg~\max_{q_{n}}~\mathfrak{T}_{\mathrm{eff}}.
\end{equation}
As shown in the proof of Lemma \ref{lemma005}, since the optimum
threshold value is much larger than one, the optimizer $\hat{q}_{n}$
is sufficiently small, i.e., $\hat{q}_{n}=o(1)$.

\begin{thm} \label{th2}\label{Th1_Ch5}
Assuming the Poisson packet arrival process and large values of $n$,
the optimum solution for (\ref{opti01}) is obtained as
\begin{equation}\label{ch4: qPAP1}
q^{PAP}_{n} = \delta \dfrac{ \log ^{2} n }{n}
\end{equation}
for some constant $\delta$. Furthermore, the maximum effective
throughput of the network asymptotically scales as $\frac{\log
n}{\hat{\alpha}}$, for $\lambda = o \left( \frac{n}{\log n}
\right)$.
\end{thm}
\begin{proof}
See Appendix \ref{append02}.
\end{proof}

\begin{thm} \label{th3}\label{Th2_Ch5}
Assuming the Bernoulli packet arrival process and large values of
$n$, the optimum solution for (\ref{opti01}) is obtained as
\begin{equation}\label{ch4: qBAP1}
q^{BAP}_{n} = \delta \dfrac{\log ^{2} n }{n}
\end{equation}
for some constant $\delta$. Furthermore, the maximum effective
throughput of the network asymptotically scales as $\frac{\log
n}{\hat{\alpha}}$, for $\lambda = o \left( \frac{n}{\log n}
\right)$.
\end{thm}
\begin{proof}
See Appendix \ref{append03}.
\end{proof}

\begin{thm} \label{th4}\label{Th3_Ch5}
Assuming  a deterministic packet arrival process, the optimum
solution of (\ref{opti01}) and the corresponding maximum effective
throughput of the network are asymptotically obtained as

i) $q^{CAP}_{n} = \delta \frac{\log ^{2} n}{n}$ and $\mathfrak{T}_{\mathrm{eff}} \approx \frac{\log n}{\hat{\alpha}}$, for $\lambda =o\left(\frac{n}{\log^{2} n} \right)$,

ii) $q^{CAP}_{n} = \delta' \frac{\log ^{2} n}{n}$ and $\mathfrak{T}_{\mathrm{eff}}\approx \frac{\log n}{\hat{\alpha}}$, for $\lambda =\Theta \left(\frac{n}{\log^{2} n} \right)$,

iii) $q^{CAP}_{n}=\frac{ \log \left(\frac{\lambda \log^{2} \lambda}{n \hat{\alpha}} \right)}{\lambda}$ and $\mathfrak{T}_{\mathrm{eff}}\approx \frac{\log n}{\hat{\alpha}}$, for $\lambda =\omega\left(\frac{n}{\log^{2} n} \right)$ and $\lambda=o \left(\frac{n}{\log n} \right)$,

for some constants $\delta$ and $\delta'$.
\end{thm}
\begin{proof}
See Appendix \ref{append04}.
\end{proof}


The above theorems imply that the effective throughput of the
network scales as $\frac{\log n}{\hat{\alpha}}$, regardless of the
packet arrival process. Note that this value is the same as the
sum-rate scaling of the same network with backlogged users
\cite{JamshidITIT2008}, which is an upper-bound on the effective
throughput of the current setup. In other words, the effect of the
real-time traffic in the throughput (which is captured in the full
buffer probability) is asymptotically negligible. However, we did
not consider the effect of dropping on the calculations. In the
subsequent section, we include the dropping probability in the
analysis and find the maximum effective throughput of the network
such that the dropping probability approaches zero.

\section{Delay Analysis} \label{secdelay1}
In this section, we first formulate the packet dropping probability
in the underlying network in terms of the number of links ($n$) and
$\lambda$ for the aforementioned packet arrival processes. Then, we
derive the sufficient conditions for the delay-bound ($\lambda$) in
the asymptotic case of $n \to \infty$ such that the packet dropping
probabilities tend to zero, while achieving the maximum effective
throughput of the network.

\begin{lem} \label{lem003} Let us denote the packet dropping probability of a link $i$,
$i \in \mathbb{N}_{n}$, for the Poisson, Bernoulli and constant arrival processes as
$\mathbb{P} \left \{\mathscr{B}^{PAP}_{i} \right \}$, $\mathbb{P} \left \{\mathscr{B}^{BAP}_{i} \right \}$
and $\mathbb{P} \left \{\mathscr{B}^{CAP}_{i} \right \}$, respectively. Then,
\begin{eqnarray}
\label{eqpoisson1} \mathbb{P} \left \{\mathscr{B}^{PAP}_{i} \right \}& =& \dfrac{1}{1+\lambda \log (1-q_{n})^{-1}},\\
\label{eqBernou1} \mathbb{P} \left \{\mathscr{B}^{BAP}_{i} \right \}& =&\dfrac{(1-q_{n})(\lambda q_{n})^{-1}}{1+(1-q_{n})(\lambda q_{n})^{-1}},\\
\label{eq001} \mathbb{P} \left \{\mathscr{B}^{CAP}_{i} \right \} &=& (1-q_{n})^{\lambda}.
\end{eqnarray}
\end{lem}
\begin{proof}
Recalling $t_{A_{k}}^{(i)}$ as the time instant of the $k^{th}$
packet arrival into the buffer of link $i$, each user $i$ is active
at time slot $t \geq t_{A_{k}}^{(i)}$ only when
$h^{(t)}_{ii}>\tau_{n}$. In other words, assuming the buffer is
full, no transmission (or no service) occurs in each slot with
probability $1-q_{n}$. From (\ref{inter01}) and (\ref{eqn:
event01})-(\ref{eqn: BAP1}), since the time duration  between
subsequent packet arrivals is $x^{(i)}_{k}$, the packet dropping
probability for a link $i$ is obtained as
\begin{equation}\label{linkbl}
\mathbb{P} \left \{\mathscr{B}_{i} \right \} =\mathbb{E} \left[( 1-q_{n})^{x^{(i)}_{k}} \right],
\end{equation}
where the expectation is computed with respect to $x^{(i)}_{k}$. For
the PAP, since $x^{(i)}_{k}$ is an exponential random variable,
(\ref{linkbl}) can be simplified as
\begin{eqnarray}
 \mathbb{P} \left \{\mathscr{B}^{PAP}_{i} \right \} &=& \int_{0}^{\infty}\dfrac{1}{\lambda}(1-q_{n})^{x}e^{-\frac{1}{\lambda}x} dx\\
\label{dropping2}&=& \dfrac{1}{1+\lambda \log (1-q_{n})^{-1}}.
\end{eqnarray}

Also for the BAP, $x^{(i)}_{k}$ is a geometric random variable with
parameter $\rho=\dfrac{1}{\lambda}$. Thus, (\ref{linkbl}) can be
simplified as
\begin{eqnarray}
\label{eqn: BAP2}\mathbb{P} \left \{\mathscr{B}^{BAP}_{i} \right \} &=& \sum_{m=1}^{\infty}(1-q_{n})^{m}\rho(1-\rho)^{m-1}\\
&=& \dfrac{\rho}{1-\rho}\sum_{m=1}^{\infty}\left [(1-q_{n})(1-\rho)\right]^{m}\\
\label{dropping3}&\stackrel{(a)}{=} & \dfrac{(1-q_{n})(\lambda q_{n})^{-1}}{1+(1-q_{n})(\lambda q_{n})^{-1}},
\end{eqnarray}
where $(a)$ comes from the following geometric series:
\begin{equation}
\sum_{m=1}^{\infty}x^m=\dfrac{x}{1-x},~~~~~~~|x|<1.
\end{equation}

According to Fig. \ref{fig: Delay1}-a, $x^{(i)}_{k}$ for the CAP is
a deterministic quantity and is equal to $\lambda$. Thus, we have
\begin{equation}\label{block02}
\mathbb{P} \left \{\mathscr{B}^{CAP}_{i} \right \} = (1-q_{n})^{\lambda}.
\end{equation}

It should be noted that (\ref{dropping2}), (\ref{dropping3}) and
(\ref{block02}) are valid for every value of $q_{n} \in [0,1]$. In
particular, in the extreme case of $q_{n}=1$, $\mathbb{P} \left
\{\mathscr{B}^{CAP}_{i} \right \}=\mathbb{P} \left
\{\mathscr{B}^{PAP}_{i} \right \}=\mathbb{P} \left
\{\mathscr{B}^{BAP}_{i} \right \}=0$.
\end{proof}

We are now ready to prove the main result of this section. In the
next theorem, we derive the sufficient conditions on  $\lambda$,
such that the corresponding packet dropping probabilities tend to
zero, while achieving the maximum effective throughput of the
network.

\begin{thm} \label{th001} For the optimum $q_{n}$ obtained in Theorems \ref{th2}-\ref{th4}
resulting in the maximum effective throughput of the network,

i) $\lim_{n \to \infty} \mathbb{P} \left \{\mathscr{B}^{PAP}_{i}\right\}=0$, if
$\lambda^{PAP} = \omega \left(\frac{n}{\log^{2} n} \right)$ and
$\lambda^{PAP} = o \left(\frac{n}{\log n} \right)$,

ii) $\lim_{n \to \infty} \mathbb{P} \left \{\mathscr{B}^{BAP}_{i}\right\}=0$, if
$\lambda^{BAP} = \omega \left(\frac{n}{\log^{2} n} \right)$ and
$\lambda^{BAP} = o \left(\frac{n}{\log n} \right)$,

iii) $\lim_{n \to \infty} \mathbb{P} \left \{\mathscr{B}^{CAP}_{i}\right\}=0$, if $\lambda^{CAP} = \omega \left(\frac{n}{\log^{2} n} \right)$ and $\lambda^{CAP} = o \left(\frac{n}{\log n} \right)$.

\end{thm}
\begin{proof}
i) From (\ref{eqpoisson1}), we have
\begin{eqnarray}
\label{drop_PAP1}\mathbb{P}\left\{\mathscr{B}^{PAP}_{i} \right\} &=& \dfrac{1}{1-\lambda^{PAP} \log (1-q^{PAP}_{n})}.
\end{eqnarray}
It follows from (\ref{drop_PAP1}) that achieving
$\mathbb{P}\left\{\mathscr{B}^{PAP}_{i}\right\}=\epsilon$ results in
\begin{eqnarray}\label{PAP_epsilon1}
\lambda^{PAP}_{\epsilon} &=& \dfrac{1-\epsilon^{-1}}{\log (1-q^{PAP}_{n})} \notag\\
&\stackrel{(a)}{\approx}& \frac{\epsilon^{-1}-1}{q^{PAP}_{n}},
\end{eqnarray}
where $(a)$ comes from $q^{PAP}_{n}=o(1)$ and the approximation
$\log(1-z) \approx -z, \quad |z| \ll 1$. Noting the fact that the
optimum value of $q^{PAP}_{n}$ scales as $\Theta \left( \frac{\log^2
n}{n}\right)$, having $\lambda^{PAP} = \omega
\left(\frac{n}{\log^{2} n} \right)$ results in $\lim_{n \to \infty}
\mathbb{P}\left\{\mathscr{B}^{PAP}_{i}\right\}=0$. On the other
hand, from Theorem \ref{Th1_Ch5}, the condition $\lambda^{PAP} = o
\left(\frac{n}{\log n} \right)$ is required to achieve the maximum
$\mathfrak{T}_{\mathrm{eff}}$, and this completes the proof of the
first part of the Theorem.

ii) It is realized from (\ref{eqBernou1}) that achieving
$\mathbb{P}\left\{\mathscr{B}^{BAP}_{i}\right\}=\epsilon$ results in
\begin{eqnarray}\label{BAP_epsilon1}
\lambda^{BAP}_{\epsilon}&=&\dfrac{1}{q^{BAP}_{n}}\left[ (1-q^{BAP}_{n})\epsilon^{-1}-(1-q^{BAP}_{n}) \right] \notag\\
&\approx& \frac{\epsilon^{-1}}{q^{BAP}_{n}},
\end{eqnarray}
for small enough $\epsilon$. Noting the fact that the optimum value
of $q^{BAP}_{n}$ scales as $\Theta \left( \frac{\log^2
n}{n}\right)$, having $\lambda^{BAP} = \omega
\left(\frac{n}{\log^{2} n} \right)$ results in $\lim_{n \to \infty}
\mathbb{P}\left\{\mathscr{B}^{BAP}_{i}\right\}=0$. On the other
hand, from Theorem \ref{Th2_Ch5}, $\lambda^{BAP} = o
\left(\frac{n}{\log n} \right)$ guarantees achieving the maximum
effective throughput of the network.

iii) From (\ref{eq001}), we have
\begin{eqnarray}
\mathbb{P}\left\{\mathscr{B}^{CAP}_{i}\right\} &=& e^{\lambda^{CAP}\log(1-q^{CAP}_{n})}\\
\label{e3}&\stackrel{(a)}{\approx} & e^{-q^{CAP}_{n}\lambda^{CAP}}
\end{eqnarray}
where $(a)$ follows from $\log(1-z) \approx -z, \quad |z| \ll 1$ for
$q^{CAP}_{n} = o(1)$. To achieve
$\mathbb{P}\left\{\mathscr{B}^{CAP}_{i}\right\}=\epsilon$, we must
have
\begin{equation}\label{CAP_epsilon1}
\lambda^{CAP}_{\epsilon}=\dfrac{1}{q^{CAP}_{n}}\log \epsilon^{-1}.
\end{equation}
It follows from (\ref{e3}) that setting
$q^{CAP}_{n}\lambda^{CAP}=\omega(1)$ makes
$e^{-q^{CAP}_{n}\lambda^{CAP}} \rightarrow 0$. Using part $(iii)$ in
Theorem \ref{Th3_Ch5}, it turns out that choosing $\lambda^{CAP}=
\omega \left(\frac{n}{\log^{2} n} \right)$ satisfies
$q^{CAP}_{n}\lambda^{CAP}=\omega(1)$ which yields $\lim_{n \to
\infty} \mathbb{P}\left\{\mathscr{B}^{CAP}_{i}\right\}=0$. We also
need the condition $\lambda^{CAP} = o \left(\frac{n}{\log n}
\right)$ to ensure achieving the maximum effective throughput of the
network.
\end{proof}

\textit{Remark 1-} It is worth mentioning that the delay-bound
($\lambda$) in each link for the CAP scales the same as that of for
the PAP and BAP. However,
$\mathbb{P}\left\{\mathscr{B}^{CAP}_{i}\right\}$ decays faster than
$\mathbb{P}\left\{\mathscr{B}^{PAP}_{i}\right\}$ and
$\mathbb{P}\left\{\mathscr{B}^{BAP}_{i}\right\}$ in terms of
$\lambda$, when $n$ tends to infinity (exponentially versus
linearly).

An interesting conclusion of Theorem \ref{th001} is the possibility
of achieving the maximum effective throughput of the network while
making the dropping probability approach zero. More precisely, there
exists some $\epsilon \ll 1$ such that
$\mathbb{P}\left\{\mathscr{B}_{i}\right\} \leq \epsilon$, $\forall i
\in \mathbb{N}_n$, while achieving the maximum
$\mathfrak{T}_{\rm{eff}}$ of $\frac{\log n}{\hat{\alpha}}$. This is
true for all aforementioned arrival processes. However, for
arbitrary values of $\epsilon$, there is a tradeoff between
increasing the throughput, and decreasing the dropping probability
and the delay-bound ($\lambda$). This tradeoff is studied in the
next section.


\section{Throughput-Delay-Dropping Probability Tradeoff}\label{TDT}
In this section, we study the tradeoff between the effective
throughput of the network and other performance measures, i.e., the
dropping probability and the delay-bound ($\lambda$) for different
packet arrival processes. In particular, we would like to know how
much degradation will be enforced in the throughput by introducing
the other constraints, and how much this degradation depends on the
packet arrival process.

\subsection{Tradeoff Between Throughput and Dropping Probability}
In this section, we assume that  a constraint
$\mathbb{P}\left\{\mathscr{B}_{i}\right\} \leq \epsilon$ must be
satisfied for the dropping probability. It can be easily shown that
the constraint $\mathbb{P}\left\{\mathscr{B}_{i}\right\} \leq
\epsilon$ is equivalent to $\mathbb{P}\left\{\mathscr{B}_{i}\right\}
= \epsilon$. The aim is to characterize the degradation on the
effective throughput of the network in terms of $\epsilon$ for
different packet arrival processes. First, we consider PAP.

Looking at the equations (\ref{PAP001}) and (\ref{eqpoisson1}), it
turns out that $\mathbb{P}\left\{\mathscr{B}_{i}^{PAP}\right\} =
\Delta_n^{PAP}$. Hence, the condition
$\mathbb{P}\left\{\mathscr{B}_{i}^{PAP}\right\} = \epsilon$ is
translated to $ \Delta_n^{PAP} = \epsilon$. Therefore, using
(\ref{opio}), the effective throughput of the network can be written
as
\begin{eqnarray}
 \mathfrak{T}_{\rm{eff}} \approx n q_n \epsilon \log \left( 1+ \frac{\tau_n}{n \hat{\alpha} q_n \epsilon} \right).
\end{eqnarray}
From the above equation, it can be realized that the effective
throughput of the network is equal to the average sum-rate of the
network with $n \epsilon$ users in the case of backlogged users,
which is given  in \cite{JamshidITIT2008} as $\frac{\log (n
\epsilon)}{\hat{\alpha}}$ for the case of $n \epsilon \gg 1$ or
$\epsilon = \omega (\frac{1}{n})$. Also, the optimum value of $q_n$
is shown to scale as $\delta  \frac{\log^2(n \epsilon)}{n \epsilon}$
for some constant $\delta$ and hence, the optimum value of $\lambda$
is given as $\frac{\epsilon^{-1}}{q_n} =\frac{n}{\delta \log^2(n
\epsilon)}$. Let us denote $\Delta \mathfrak{T}_{\rm{eff}}$ as the
degradation in the effective throughput of the network, which is
defined as the difference between the maximum effective throughput
in the case of no constraint on
$\mathbb{P}\left\{\mathscr{B}_{i}\right\}$ (Theorem
\ref{th2}-\ref{th4}) and the case with constraint on
$\mathbb{P}\left\{\mathscr{B}_{i}\right\}$. Using Theorem \ref{th2},
$\Delta \mathfrak{T}_{\rm{eff}}$ for the PAP can be written as
\begin{eqnarray}
 \Delta \mathfrak{T}_{\rm{eff}} &\approx& \frac{\log n}{\hat{\alpha}} - \frac{\log (n \epsilon)}{\hat{\alpha}} \notag\\
&=& \frac{\log (\epsilon^{-1})}{\hat{\alpha}}, \label{dool1}
\end{eqnarray}
for $\epsilon = \omega \left( \frac{1}{n} \right)$\footnote{In the
case of $\epsilon = O (\frac{1}{n})$, it is easy to see that the
effective throughput of the network does not scale with $n$.}.
Moreover, for values of $\epsilon$ such that $\log (\epsilon^{-1}) =
o (\log n)$, it can be shown that the scaling of the effective
throughput of the network is not changed, i.e.,
$\mathfrak{T}_{\rm{eff}} \sim \frac{\log n}{\hat{\alpha}} $.

For the BAP, and using (\ref{BAP001}) and (\ref{eqBernou1}), we have
\begin{eqnarray}
 \mathbb{P}\left\{\mathscr{B}_{i}^{BAP}\right\} &=& \frac{1-q_n}{1+(\lambda-1)q_n} \notag\\
&\stackrel{(a)}{\approx}& \frac{1}{1+(\lambda-1)q_n} \notag\\
&=& \Delta^{BAP}_n,
\end{eqnarray}
where $(a)$ follows from the fact that $q_n = o(1)$. Therefore,
similar to the case of the PAP, we have
$\mathbb{P}\left\{\mathscr{B}_{i}^{BAP}\right\} \approx
\Delta^{BAP}_n =\epsilon$ and as a result, the rest of the arguments
hold. In particular,
\begin{eqnarray}
\Delta \mathfrak{T}_{\rm{eff}} &\approx& \frac{\log (\epsilon^{-1})}{\hat{\alpha}}.
\end{eqnarray}

For the CAP, and using (\ref{CAP001}) and (\ref{eq001}), we have
\begin{eqnarray}
(1-q_n)^{\lambda} = \epsilon \quad \Longrightarrow \lambda q_n \approx \log (\epsilon^{-1}),
\end{eqnarray}
which gives
\begin{eqnarray}
\Delta^{CAP}_n &=& \frac{1-(1-q_n)^{\lambda}}{\lambda q_n} \\
&\approx&  \frac{1}{\log (\epsilon^{-1})}.
\end{eqnarray}
Hence, using (\ref{opio}), the effective throughput of the network
can be written as
\begin{eqnarray}
\mathfrak{T}_{\rm{eff}} \approx \frac{n}{\log (\epsilon^{-1})} q_n  \log \left( 1+ \frac{\tau_n}{\frac{n}{\log (\epsilon^{-1})} \hat{\alpha} q_n } \right),
\end{eqnarray}
which is equal to the average sum-rate of a network with
$\frac{n}{\log (\epsilon^{-1})}$ backlogged users and is
asymptotically equal to $\frac{\log \left( \frac{n}{\log
(\epsilon^{-1})}\right)}{\hat{\alpha}}$, for values of $\epsilon$
satisfying $\log (\epsilon^{-1}) = o(n)$. Therefore, the degradation
in the effective throughput of the network for the CAP can be
expressed as
\begin{eqnarray}
\Delta \mathfrak{T}_{\rm{eff}} &\approx& \frac{\log n}{\hat{\alpha}} - \frac{\log \left( \frac{n}{\log (\epsilon^{-1})}\right)}{\hat{\alpha}} \notag\\
&=& \frac{\log \log (\epsilon^{-1})}{\hat{\alpha}}. \label{dool3}
\end{eqnarray}
Comparing the expressions of $\Delta \mathfrak{T}_{\rm{eff}}$ for
the Poisson, Bernoulli and constant packet arrival processes, it
follows that the degradation in the effective throughput of the
network in the cases of PAP and BAP both grow logarithmically with
$\epsilon^{-1}$, while in the case of CAP it grows double
logarithmically. In other words, the degradation in the throughput
in the cases of the PAP and BAP is much more substantial compared to
the CAP. This fact is also observed in the simulation results in the
next section.

\subsection{Tradeoff Between Throughput and Delay}
In this section, we aim to find the tradeoff between the effective
throughput of the network and the delay-bound ($\lambda$), for a
given constraint on the dropping probability, i.e.,
$\mathbb{P}\left\{\mathscr{B}_{i}\right\} \leq \epsilon$.

\subsubsection{PAP}
Using (\ref{PAP001}) and (\ref{eqpoisson1}), it follows that for a
given $\lambda$ and $\epsilon \ll 1$, we have
\begin{eqnarray}
 q_n &\approx& \frac{\epsilon^{-1}}{\lambda}, \notag\\
\Longrightarrow \tau_n &\approx& \log (\lambda \epsilon),
\end{eqnarray}
and
\begin{eqnarray}
q_n \Delta_n &\approx& \frac{1}{\lambda}.
\end{eqnarray}
Substituting $q_n \Delta_n$ and $\tau_n$ from the above equations in
(\ref{opio}) yields
\begin{eqnarray}
 \mathfrak{T}_{\rm{eff}} \approx \frac{n}{\lambda} \log \left(1+ \frac{\lambda \log (\lambda \epsilon)}{n \hat{\alpha}} \right).
\end{eqnarray}
It can be verified that $ \mathfrak{T}_{\rm{eff}}$ has a global
maximum at $\lambda^{PAP}_{opt} \approx \frac{n \hat{\alpha}}{\log^2
(n \hat{\alpha} \epsilon^{-1})}$. In other words, for $\lambda <
\lambda^{PAP}_{opt}$, there is a tradeoff between the throughput and
delay, meaning that increasing $\lambda$ results in increasing both
the throughput and delay. However, the increase in the throughput is
logarithmic while the delay increases linearly with $\lambda$. It
should be noted that the region $\lambda> \lambda^{PAP}_{opt}$ is
not of interest, since increasing $\lambda$ from
$\lambda^{PAP}_{opt}$ results in decreasing the throughput and
increasing the delay which is not desired.

\subsubsection{BAP} Due to the similarity between the values of $\mathbb{P}\left\{\mathscr{B}_{i}\right\}$
and $\Delta_n$ for the PAP and the BAP, the results obtained for the PAP are also valid for the BAP.

\subsubsection{CAP} Using (\ref{CAP001}) and (\ref{eq001}), it follows that for a given $\lambda$ and
$\epsilon \ll 1$, we have
\begin{eqnarray}
 q_n &\approx& \frac{\log (\epsilon^{-1})}{\lambda}, \notag\\
\Longrightarrow \tau_n &\approx& \log \left(\frac{\lambda}{ \log (\epsilon^{-1})} \right),
\end{eqnarray}
and
\begin{eqnarray}
q_n \Delta_n &\approx& \frac{1}{\lambda}.
\end{eqnarray}
As can be observed, all the results for the cases of PAP and BAP are
extendable to the case of CAP by substituting $\epsilon^{-1}$ with
$\log (\epsilon^{-1})$. In particular, the optimum value for
$\lambda$ can be written as $\lambda^{CAP}_{opt} \approx \frac{n
\hat{\alpha}}{\log^2 \left(n \hat{\alpha} \log (\epsilon^{-1})
\right)}$, and for $\lambda < \lambda^{CAP}_{opt}$, the effective
throughput of the network can be given as $\mathfrak{T}_{\rm{eff}}
\approx \frac{1}{\hat{\alpha}} \log \left( \frac{\lambda}{\log
(\epsilon^{-1})}\right)$. This means that in the region $\lambda <
\lambda^{CAP}_{opt}$, which is the region of interest, there is a
tradeoff between the throughput and delay such that by increasing
$\lambda$, $\mathfrak{T}_{\rm{eff}}$ increases logarithmically,
while the delay increases linearly with $\lambda$. Furthermore,
comparing the value of $\lambda_{opt}$ for the PAP and BAP with the
CAP, it is realized that $\lambda^{CAP}_{opt} > \lambda^{PAP}_{opt}$
and $\lambda^{CAP}_{opt} > \lambda^{BAP}_{opt}$. This fact is also
observed in the simulations.

\section{Numerical Results}\label{numerical_Ch4}
In this section, we present some numerical results to evaluate the
tradeoff between the effective throughput of the network and other
performance measures, i.e., dropping probability and the delay-bound
($\lambda$) for different packet arrival processes. For this
purpose, we assume that all users in the network follow the
threshold-based on-off power allocation policy. In addition, the
shadowing effect is assumed to be \textit{lognormal} distributed
with mean $\varpi=0.5$, variance $1$ and $\alpha=0.4$. Furthermore,
we assume that $n=500$ and $N_{0} = 1$.

Figures \ref{fig: delay-throuPAP} and \ref{fig: delay-throuCAP} show
the effective throughput of the network versus $\lambda_{\epsilon}$
for the PAP, BAP and CAP and different values of $\epsilon$. It is
observed from these figures that for a given constraint on the
dropping probability (e.g., $\epsilon=0.05$), and for $\lambda <
\lambda_{opt}$, increasing $\lambda$ results in increasing both the
throughput and delay. However, the increase in the throughput is
logarithmic while the delay increases linearly with $\lambda$ as
expected. Also, increasing $\lambda$ from $\lambda_{opt}$ results in
decreasing the throughput and increasing the delay which is not
desired. Furthermore, comparing the value of $\lambda_{opt}$ for the
PAP and BAP with the CAP, it is realized that $\lambda^{CAP}_{opt} >
\lambda^{PAP}_{opt}$ and $\lambda^{CAP}_{opt} >
\lambda^{BAP}_{opt}$, as expected.
\begin{figure}[bhpt]
\centerline{\psfig{figure=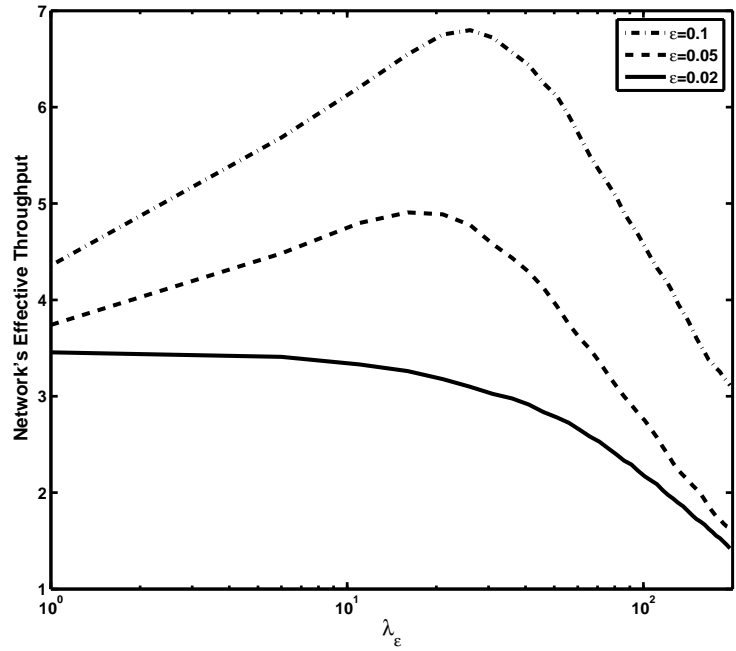,width=4.6in}}
\vspace{-7pt} \center{\hspace{16pt} \small{(a)}} \vspace{10pt}
\hspace{1pt} \centerline{\psfig{figure=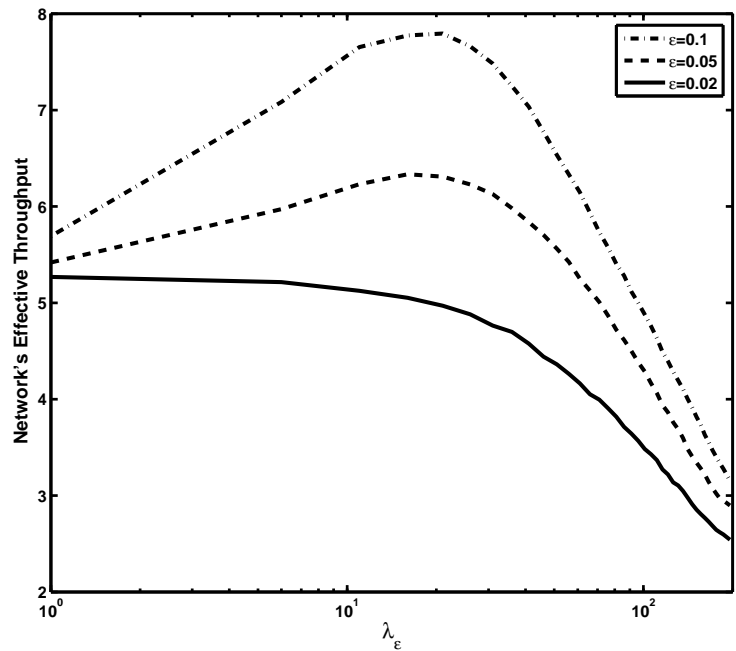,width=4.6in}}
\vspace{-35pt}
\center{\hspace{14pt} \small{(b)}} \\
\vspace{-7pt} \caption[a) $\hat{\alpha}=W=1$.] {\small{Effective
throughput of the network versus $\lambda_{\epsilon}$ for $N_{0} =
1$, $n=500$, $\alpha=0.4$, and different values of $\epsilon$ a) PAP
and b) BAP.}} \label{fig: delay-throuPAP}
\end{figure}

\begin{figure}[t]
\centerline{\psfig{figure=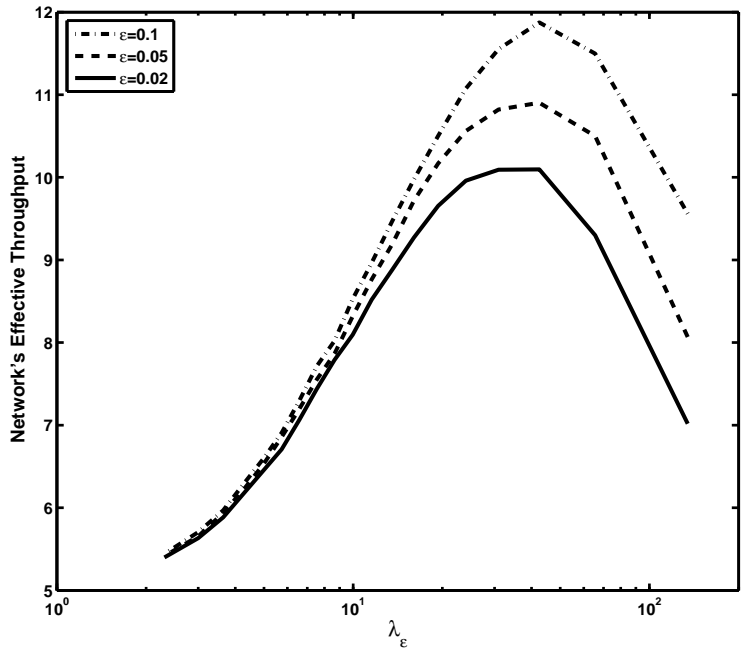,width=4.6in}}
\caption{Effective throughput of the network versus
$\lambda_{\epsilon}$ for the CAP and $N_{0} = 1$, $n=500$,
$\alpha=0.4$, and different values of $\epsilon$.} \label{fig:
delay-throuCAP}
\end{figure}

\begin{figure}[t]
\centerline{\psfig{figure=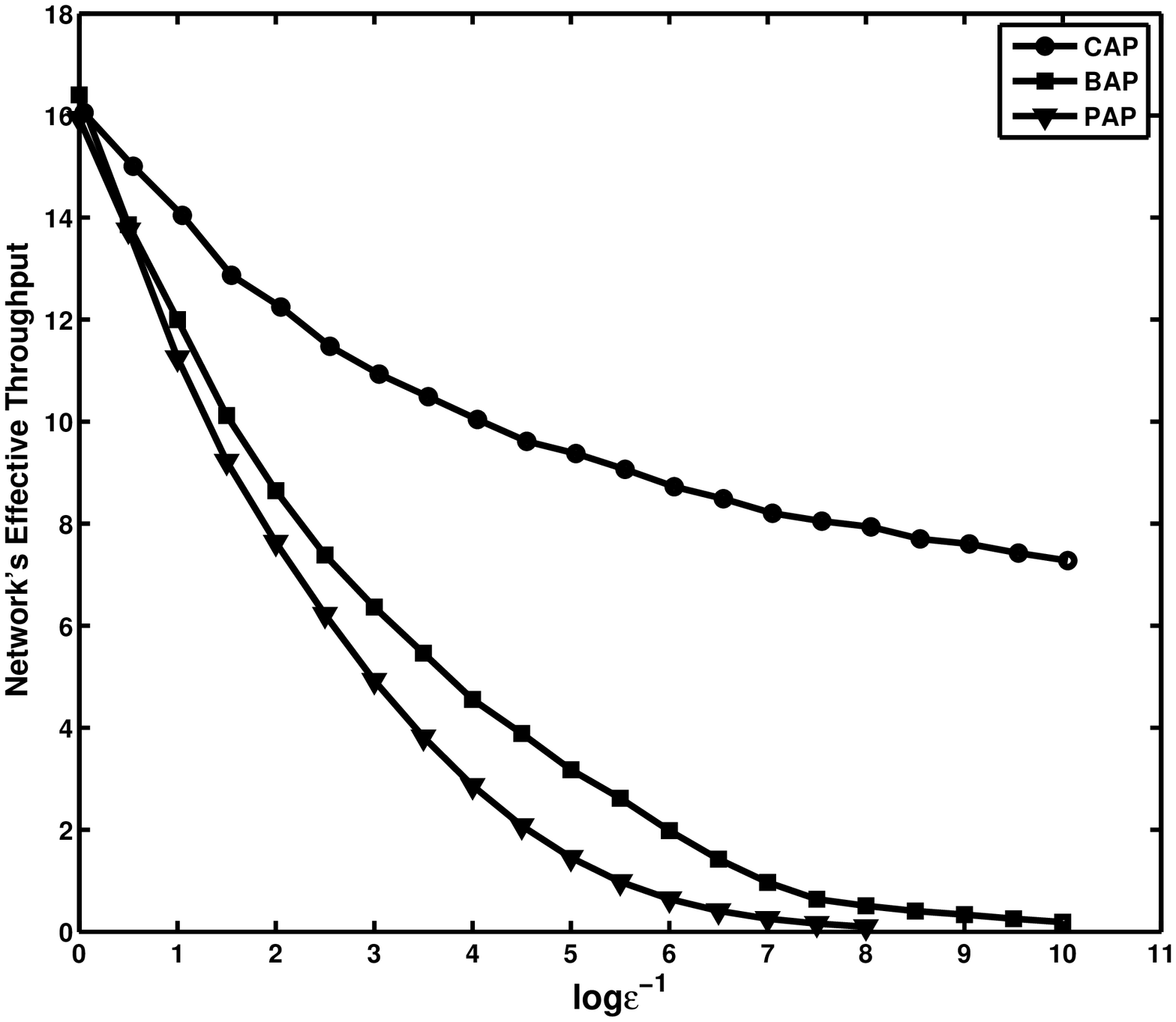,width=4.6in}}
\caption{Effective throughput of the network versus $\log
\epsilon^{-1}$ for different packet arrival processes and $N_{0} =
1$, $n=500$, $\alpha=0.4$.} \label{fig: D-T-Tradeoff}
\end{figure}

To evaluate the degradation in the effective throughput of the
network in terms of dropping probability, we plot
$\mathfrak{T}_{\rm{eff}}$ versus $\log \epsilon^{-1}$ for different
packet arrival processes in Fig. \ref {fig: D-T-Tradeoff}. It can be
seen that the degradation in the throughput in the cases of the PAP
and BAP is much more substantial compared to the CAP, as expected.
Hence, the performance of the underlying network with the CAP is
better than that of the PAP and BAP from the delay-throughput and
delay-dropping probability tradeoff points of view.

\section{Conclusion}\label{conclusion1}
In this paper, the delay-throughput of a single-hop wireless network
with $n$ links was studied. We considered a block Rayleigh fading
model with shadowing, described by parameters $(\alpha,\varpi)$, for
the channels in the network. The analysis in the paper relied on the
distributed \textit{on-off power allocation strategy} for the
deterministic and stochastic packet arrival processes. It was also
assumed that each transmitter has a buffer size of one packet and
dropping occurs once a packet arrives in the buffer while the
previous packet has not been served. In the first part of the paper,
we defined a new notion of performance in the network, called
\textit{effective throughput}, which captures the effect of arrival
process in the network throughput, and maximize it for different
cases of arrival process. It was proved that the effective
throughput of the network asymptotically scales as $\frac{\log
n}{\hat{\alpha}}$, with $\hat{\alpha} \triangleq \alpha \varpi$,
regardless of the packet arrival process. In the second part of the
paper, we presented the delay characteristics of the underlying
network in terms of the packet dropping probability. We derived the
sufficient conditions in the asymptotic case of $n \to \infty$ such
that the packet dropping probability tend to zero, while achieving
the maximum effective throughput of the network. Finally, we studied
the trade-off between the effective throughput, delay, and packet
dropping probability of the network for different packet arrival
processes. It was shown from the numerical results that the
performance of the deterministic packet arrival process is better
than that of the Poisson and the Bernoulli packet arrival processes,
from the delay-throughput and throughput-dropping probability
tradeoff points of view.

\appendices

\section{Proof of Lemma \ref{lemma000_Ch5}}\label{append01_a}

\renewcommand{\theequation}{A-\arabic{equation}}
\setcounter{equation}{0}

Let us define $\chi^{(t)}_{j} \triangleq \mathcal{L}^{(t)}_{ji}p^{(t)}_{j}$, where $\mathcal{L}^{(t)}_{ji}$
is independent of $p^{(t)}_{j}$, for $j \neq i$. Note that
\begin{eqnarray}
\mathbb{P} \left \lbrace p^{(t)}_{j}=1 \right \rbrace &=& \mathbb{P}\left \lbrace h^{(t)}_{jj}>\tau_{n},~ \mathscr{C}_{j}^{(t)}\right \rbrace\\
&\stackrel{(a)}{=} & q_{n}\Delta_{n},
\end{eqnarray}
where $(a)$ follows from (\ref{index02}). Thus for the on-off power
scheme, we have
\begin{equation}
\mathbb{E} \left[p^{(t)}_{j}\right]=q_{n}\Delta_{n}.
\end{equation}
Under a quasi-static Rayleigh fading channel model, it is concluded
that $\chi^{(t)}_{j}$'s are independent and identically distributed
(i.i.d.) random variables with
\begin{eqnarray}
\mathbb{E}\left[ \chi^{(t)}_{j} \right] & = & \mathbb{E}\left[ \mathcal{L}^{(t)}_{ji}p^{(t)}_{j} \right]=\hat{\alpha} q_{n}\Delta_{n}, \\
\textrm{Var} \left[ \chi^{(t)}_{j} \right] & = & \mathbb{E}\left[ \left(\chi^{(t)}_{j}\right)^{2} \right]-\mathbb{E}^{2}\left[ \chi^{(t)}_{j} \right]\\
&\stackrel{(a)}{\leq} & 2 \alpha \kappa q_{n}\Delta_{n}-(\hat{\alpha} q_{n}\Delta_{n})^{2},
\end{eqnarray}
where $\mathbb{E}\left[\left(h^{(t)}_{ji}\right)^{2} \right]=2$,
$\mathbb{E}\left[\left(\beta^{(t)}_{ji}\right)^{2} \right]
\triangleq \kappa$ and $\hat{\alpha} \triangleq \alpha \varpi$.
Also, $(a)$ follows from the fact that $\left(p^{(t)}_{j}\right)^{2}
\leq p^{(t)}_{j}$. Thus,
$\mathbb{E}\left[\left(p^{(t)}_{j}\right)^{2}\right] \leq
\mathbb{E}\left[p^{(t)}_{j}\right]=q_{n}\Delta_{n}$. The
interference $I^{(t)}_{i} = \sum^{n}_{\substack{j=1\\ j\neq
i}}\chi^{(t)}_{j}$ is a random variable with mean $\mu_{n}$ and
variance $\vartheta_{n}^{2}$, where
\begin{eqnarray}
\mu_{n} & \triangleq & \mathbb{E}\left[ I^{(t)}_i \right]=(n-1)\hat{\alpha} q_{n}\Delta_{n},\\
\vartheta_{n}^{2} & \triangleq & \textrm{Var} \left[I^{(t)}_i \right] \leq (n-1)(2\alpha \kappa q_{n}\Delta_{n}-(\hat{\alpha}q_{n}\Delta_{n})^{2}) \leq (n-1)(2\alpha \kappa q_{n}\Delta_{n}). \label{kir2}
\end{eqnarray}

\section{Proof of Lemma \ref{lemma001}}\label{append01}

\renewcommand{\theequation}{B-\arabic{equation}}
\setcounter{equation}{0}

Using Lemma \ref{lemma000_Ch5} and the \textit{Central Limit
Theorem} \cite[p. 183]{Petrovbook95}, we obtain
\begin{eqnarray}
\mathbb{P} \left \lbrace \vert I^{(t)}_{i}-\mu_{n} \vert <\psi_{n} \right \rbrace  &\approx&  1-Q \left(\dfrac{\psi_{n}}{\vartheta_{n}} \right) \\
&\stackrel{(a)}{\geq}& 1-e^{-\frac{\psi^2_{n}}{2 \vartheta^2_{n}}},
\end{eqnarray}
for all $\psi_{n} >0$ such that $\psi_n = o \left( n^{\frac{1}{6}}
\vartheta_n\right)$. In the above equations, the $Q(.)$ function is
defined as $Q(x)\triangleq
\frac{1}{\sqrt{2\pi}}\int_{x}^{\infty}e^{-u^{2}/2}du$, and $(a)$
follows from the fact that $Q(x) \leq e^{-\frac{x^2}{2}}$, $\forall
x> 0$. Selecting $\psi_n = \left( n
q_n\Delta_{n}\right)^{\frac{1}{8}} \sqrt{2} \vartheta_n$, we obtain
\begin{eqnarray}
 \mathbb{P} \lbrace \vert I^{(t)}_{i}-\mu_{n} \vert <\psi_{n} \rbrace &\geq& 1-e^{-\left( n q_n \Delta_n \right)^{\frac{1}{4}}}.
\end{eqnarray}
Therefore, defining $\varepsilon \triangleq \frac{\psi_n}{\mu_n}$, noting that as $\vartheta_n = O (nq_n \Delta_n)$ (from (\ref{kir2}) in Appendix I) and $\mu_n = \Theta (n q_n \Delta_n)$, we have $\varepsilon = O \left( (nq_n\Delta_n)^{-\frac{3}{8}} \right)$, it reveals that
\begin{eqnarray}
\mathbb{P} \lbrace \mu_n \left( 1-\varepsilon\right) \leq I^{(t)}_i \leq \mu_n \left( 1 + \varepsilon\right)\rbrace &\geq& 1-e^{-\left(n q_n\Delta_n\right)^{\frac{1}{4}}}.
\end{eqnarray}
Noting that $n q_n \Delta_n \to \infty$, it follows that $I^{(t)}_i
\sim \mu_n$, with probability one.

\section{Proof of Theorem \ref{th2}}\label{append02}

\renewcommand{\theequation}{C-\arabic{equation}}
\setcounter{equation}{0} Taking the first-order derivative of
(\ref{opio}) with respect to $\tau_{n}$ yields
\begin{eqnarray}
\dfrac{\partial \mathfrak{T}_{\mathrm{eff}}}{\partial \tau_{n}}&\stackrel{(a)}{=}&n q_{n} \left[\dfrac{\partial \Delta_{n}}{\partial \tau_{n}}-\Delta_{n}  \right]\log \left( 1+\dfrac{\tau_n}{n \hat{\alpha} q_{n}\Delta_{n}}\right)+nq_{n}\dfrac{(1+\tau_{n})\Delta_{n}-\tau_{n}\frac{\partial \Delta_{n}}{\partial \tau_{n}}}{n\hat{\alpha} q_{n}\Delta_{n}+\tau_{n}}\\
\label{deriv1}&\stackrel{(b)}{\approx}& n q_{n} \left[\dfrac{\partial \Delta_{n}}{\partial \tau_{n}}-\Delta_{n}  \right]\dfrac{\tau_n}{n \hat{\alpha} q_{n}\Delta_{n}}+nq_{n}\dfrac{(1+\tau_{n})\Delta_{n}-\tau_{n}\frac{\partial \Delta_{n}}{\partial \tau_{n}}}{n\hat{\alpha} q_{n}\Delta_{n}+\tau_{n}},
\end{eqnarray}
where $(a)$ comes from $q_n=e^{-\tau_n}$ and $\frac{\partial
q_{n}}{\partial \tau_{n}}=-q_{n}$. Also, $(b)$ follows from Lemma
\ref{lemma006} and using the approximation $\log(1+x) \approx x$,
for $x \ll 1$. Setting (\ref{deriv1}) equal to zero yields
\begin{equation}\label{result01}
n \hat{\alpha}q_{n} \Delta_{n}^{2}=\left(\Delta_{n}-\dfrac{\partial \Delta_{n}}{\partial \tau_{n}} \right)\tau_{n}^{2}.
\end{equation}
It should be noted that (\ref{result01}) is valid for every packet
arrival process. Recalling from (\ref{PAP001}), the full buffer
probability for the PAP is given by
\begin{eqnarray}
\Delta^{PAP}_{n} &=&\dfrac{1}{1+\lambda \log (1-q_{n})^{-1}} \\
&\stackrel{(a)}{\approx}& \frac{1}{1+\lambda q_{n}}, \label{kir1}
\end{eqnarray}
where $(a)$ follows from the fact that for $q_{n}=o(1)$,
$\log(1-q_{n})^{-1} \approx q_{n}$. In this case, $\frac{\partial
\Delta^{PAP}_{n}}{\partial \tau_{n}}=\frac{\partial
\Delta^{PAP}_{n}}{\partial q_{n}}\frac{\partial q_{n}}{\partial
\tau_{n}}=\frac{\lambda q_{n}}{\left(1+\lambda q_{n} \right)^2}$,
which results in
\begin{equation}
\Delta^{PAP}_{n}-\dfrac{\partial \Delta^{PAP}_{n}}{\partial \tau_{n}} \approx \dfrac{1}{\left(1+\lambda q_{n} \right)^2}=\left(\Delta_{n}^{PAP}\right)^2.
\end{equation}
Thus for the Poisson arrival process, (\ref{result01}) can be
simplified as
\begin{equation}\label{eqnPAP01}
n \hat{\alpha}q_{n} = \tau_{n}^2.
\end{equation}
It can be verified that the solution for (\ref{eqnPAP01}) is
\begin{equation}\label{eqnPAP02}
\tau_{n}^{PAP}=\log n-2\log \log n +O(1).
\end{equation}
Using $q_{n}=e^{-\tau_{n}}$, we conclude that
\begin{equation}\label{eqn: PAP15}
q_{n}^{PAP} = \delta \dfrac{ \log ^{2} n }{n},
\end{equation}
for some constant $\delta$.

To satisfy the condition of lemma \ref{lemma006}, we should have
\begin{equation}
 \frac{\tau_n}{n \hat{\alpha} q_n \Delta^{PAP}_{n}} \ll 1,
\end{equation}
Using (\ref{kir1}), (\ref{eqnPAP02}), and (\ref{eqn: PAP15}), it yields
\begin{equation}\label{Ch4: landaPAP1}
 \lambda^{PAP}=o\left(\dfrac{n}{\log n} \right).
\end{equation}
Thus, the maximum effective throughput of the network obtained in
(\ref{opio}) can be written as
\begin{equation}
\mathfrak{T}_{\mathrm{eff}} \approx \dfrac{\tau_{n}}{\hat{\alpha}}.
\end{equation}

\section{Proof of Theorem \ref{th3}}\label{append03}

\renewcommand{\theequation}{D-\arabic{equation}}
\setcounter{equation}{0}

Using (\ref{BAP001}), we have
$\frac{\partial \Delta^{BAP}_{n}}{\partial \tau_{n}}=\frac{\partial \Delta^{BAP}_{n}}{\partial q_{n}}\frac{\partial q_{n}}{\partial \tau_{n}}
=-q_{n}\frac{\partial \Delta^{BAP}_{n}}{\partial q_{n}}=\frac{q_{n}(\lambda-1)}{\left(1+(\lambda-1)q_{n} \right)^2}$.
In this case,
\begin{equation}
\Delta^{BAP}_{n}-\dfrac{\partial \Delta^{BAP}_{n}}{\partial \tau_{n}}=\dfrac{1}{\left(1+(\lambda-1)q_{n} \right)^2}=\left(\Delta_{n}^{BAP}\right)^2.
\end{equation}
Thus for the Bernoulli arrival process, (\ref{result01}) can be
simplified as
\begin{equation}\label{eqnBAP01}
n \hat{\alpha}q_{n} = \tau_{n}^2.
\end{equation}
It can be observed that (\ref{eqnBAP01}) is exactly equal to
(\ref{eqnPAP01}) and hence, its solution can be written as
\begin{equation}\label{eqnBAP02}
\tau_{n}^{BAP}=\log n-2\log \log n +O(1),
\end{equation}
and
\begin{equation}\label{eqn: BAP15}
q_{n}^{BAP} = \delta \dfrac{ \log ^{2} n }{n},
\end{equation}
for some constants $\delta$. Similarly, the maximum effective
throughput of the network for the BAP is obtained as
\begin{equation}
\mathfrak{T}_{\mathrm{eff}} \approx \dfrac{\tau_{n}}{\hat{\alpha}},
\end{equation}
which is achieved under the condition
\begin{equation}\label{ch4: landaBAP}
 \lambda^{BAP}=o\left(\dfrac{n}{\log n} \right).
\end{equation}

\section{Proof of Theorem \ref{th4}}\label{append04}

\renewcommand{\theequation}{E-\arabic{equation}}
\setcounter{equation}{0}

Using (\ref{CAP001}), we have
\begin{eqnarray}
\dfrac{\partial \Delta^{CAP}_{n}}{\partial \tau_{n}}&=&\dfrac{\partial \Delta^{CAP}_{n}}{\partial q_{n}}\dfrac{\partial q_{n}}{\partial \tau_{n}}\\
&=&-q_{n}\dfrac{\partial \Delta^{CAP}_{n}}{\partial q_{n}}\\
&=&\dfrac{1-(1-q_{n})^{\lambda}}{\lambda q_{n}}-(1-q_{n})^{\lambda-1}\\
&=&\Delta^{CAP}_{n}-(1-q_{n})^{\lambda-1}.
\end{eqnarray}
Hence, $\Delta^{CAP}_{n}-\frac{\partial \Delta^{CAP}_{n}}{\partial \tau_{n}}=(1-q_{n})^{\lambda-1}$.
In this case, (\ref{result01}) can be simplifies as
\begin{equation}\label{result02}
n \hat{\alpha}q_{n} \dfrac{\left[1-(1-q_{n})^{\lambda}\right]^2}{\left(\lambda q_{n}\right)^2}=(1-q_{n})^{\lambda-1}\tau_{n}^{2}.
\end{equation}
or
\begin{equation}\label{result03}
n \hat{\alpha}=\dfrac{\tau_{n}^{2}\lambda^2 q_{n}(1-q_{n})^{\lambda-1}}{\left[1-(1-q_{n})^{\lambda}\right]^2}.
\end{equation}

Since $q_{n} =o(1)$, we have
$(1-q_{n})^{\lambda-1}=e^{(\lambda-1)\log(1-q_{n})}
\stackrel{(a)}{\approx} e^{-\lambda q_{n}}$, and
$1-(1-q_{n})^{\lambda} \stackrel{(b)}{\approx} 1-e^{-\lambda
q_{n}}$. It should be noted that $(a)$ and $(b)$ are valid under the
condition $\frac{\lambda q_{n}^{2}}{2}=o(1)$ \footnote{As we will
show the condition $\frac{\lambda q_{n}^{2}}{2}=o(1)$ is satisfied
for the optimum $q_{n}$ and the corresponding $\lambda$.}. Thus,
(\ref{result03}) can be simplified as
\begin{equation}
n \hat{\alpha}=\dfrac{\tau_{n}^{2}\lambda^2 q_{n}e^{-\lambda q_{n}}}{\left[1-e^{-\lambda q_{n}}\right]^2},
\end{equation}
or
\begin{equation}\label{eqn: Psi1}
\dfrac{\nu \log \nu^{-1}}{(1-\nu)^{2}}=\Psi,
\end{equation}
where $\nu \triangleq e^{-\lambda q_{n}}$ and $\Psi \triangleq \dfrac{n \hat{\alpha}}{\tau_{n}^{2} \lambda}$.
For this setup, we have the following cases:

\textbf{Case 1:} $\Psi \gg 1$

It is realized from (\ref{eqn: Psi1}) that for $\Psi \gg 1$,
$\nu=1-\epsilon$, where $\epsilon=o(1)$. Thus, (\ref{eqn: Psi1}) can
be simplified as
\begin{eqnarray}
\Psi &\approx& \dfrac{\log(1-\epsilon)^{-1}}{\epsilon^2}\\
&\stackrel{(a)}{\approx}&\dfrac{\epsilon}{\epsilon^2}\\
&=&\dfrac{1}{\epsilon},
\end{eqnarray}
where $(a)$ follows from the Taylor series expansion
$\log(1-z)=-\sum_{k=1}^{\infty} \dfrac{z^{k}}{k} \approx -z,~~ \vert
z \vert \ll 1$. Since $\nu \triangleq e^{-\lambda q_{n}}$ and
$\nu=1-\epsilon$, we have
\begin{eqnarray}
e^{-\lambda q_{n}}&=&1-\dfrac{1}{\Psi},\\
\Longrightarrow \label{result04}q_{n} &\stackrel{(a)}{\approx}& \dfrac{1}{\Psi \lambda}=\dfrac{\tau_{n}^{2}}{n \hat{\alpha}},
\end{eqnarray}
where $(a)$ follows from the fact that as $\lambda q_n =o(1)$, we
have $e^{-\lambda q_n} \approx 1-\lambda q_n$. It can be verified
that the solution for (\ref{result04}) is
\begin{equation}\label{eqnCAP01}
\tau_{n}^{CAP}=\log n-2\log \log n +O(1).
\end{equation}
Using $q_{n}=e^{-\tau_{n}}$, we conclude that
\begin{equation}\label{eqn: 15}
q_{n}^{CAP} = \delta \dfrac{ \log ^{2} n }{n},
\end{equation}
for some constant  $\delta$.

The above results are valid for $\Psi \triangleq \frac{n
\hat{\alpha}}{\tau_{n}^{2} \lambda} \gg1$ or $\lambda
=o\left(\frac{n}{\log^{2} n} \right)$. Also, it can be verified that
$\frac{\lambda q_{n}^{2}}{2}=o(1)$, and therefore the
approximations $(1-q_{n})^{\lambda-1} \approx e^{-\lambda q_{n}}$
and $1-(1-q_{n})^{\lambda} \approx 1-e^{-\lambda q_{n}}$ are valid
in this region.

To satisfy the condition of Lemma \ref{lemma006}, we must have
\begin{equation}\label{condi01}
 \dfrac{\tau_n}{n \hat{\alpha} q^{CAP}_n \Delta^{CAP}_{n}} \ll 1.
\end{equation}
From (\ref{CAP001}), (\ref{eqnCAP01}) and noting that as $\lambda
=o\left(\frac{n}{\log^{2} n} \right)$, $\left[ 1-(1-q_{n})^{\lambda}
\right] \approx 1-e^{-\lambda q_n} \approx \lambda q_n$, we can
write
\begin{eqnarray}
\dfrac{\tau_n}{n \hat{\alpha} q^{CAP}_n \Delta^{CAP}_{n}} &\approx& \dfrac{\lambda \log n}{n \hat{\alpha}\left[ 1-(1-q_{n})^{\lambda} \right]}\\
&\approx& \frac{\log n}{n \hat{\alpha} q_n} \notag\\
&=& O \left(\frac{1}{\log n} \right),
\end{eqnarray}
which means that the condition of Lemma \ref{lemma006} is
automatically satisfied in this region. Thus, the maximum effective
throughput of the network obtained in (\ref{opio}) can be simplified
as
\begin{equation}
\mathfrak{T}_{\mathrm{eff}} \approx \dfrac{\tau_{n}}{\hat{\alpha}} \approx \frac{\log n}{\hat{\alpha}}.
\end{equation}

\textbf{Case 2:} $\Psi =\Theta(1)$

From (\ref{eqn: Psi1}) which gives $\frac{\nu \log
\nu^{-1}}{(1-\nu)^{2}}=\Psi=\Theta(1)$, we conclude that $\nu
\triangleq e^{-\lambda q_{n}}=\Theta(1)$. Thus,
\begin{eqnarray}
q_{n}&=&\dfrac{c_{1}}{\lambda}\\
\label{result05}&\stackrel{(a)}{=}& \dfrac{c_{2}\tau_{n}^{2}}{n \hat{\alpha}}
\end{eqnarray}
where $c_{1}$ and $c_{2}$ are constants and $(a)$ follows from $\Psi
\triangleq \frac{n \hat{\alpha}}{\tau_{n}^{2} \lambda} =\Theta(1)$.
It can be verified that the solution for (\ref{result05}) is
\begin{eqnarray}
\label{eqnCAP02}\tau_{n}^{CAP}&=&\log n-2\log \log n +O(1).\\
\label{eqn: 15_1}q_{n}^{CAP} &=& \delta' \dfrac{ \log ^{2} n }{n},
\end{eqnarray}
for some constant  $\delta'$.

The above results are valid for $\Psi \triangleq \frac{n \hat{\alpha}}{\tau_{n}^{2} \lambda} =\Theta(1)$ or
$\lambda =\Theta\left(\frac{n}{\log^{2} n} \right)$. Also, it can be verified that $\frac{\lambda q_{n}^{2}}{2}=o(1)$,
and therefore, the approximations $(1-q_{n})^{\lambda-1} \approx e^{-\lambda q_{n}}$ and
$1-(1-q_{n})^{\lambda} \approx 1-e^{-\lambda q_{n}}$ are valid in this region.

Similar to the argument in Case 1, the  condition of Lemma
\ref{lemma006} is satisfied, and therefore, the maximum effective
throughput of the network is obtained as
\begin{equation}
\mathfrak{T}_{\mathrm{eff}} \approx \dfrac{\tau_{n}}{\hat{\alpha}} \approx \frac{\log n}{\hat{\alpha}}.
\end{equation}

\textbf{Case 3:} $\Psi \ll 1$

It is concluded from (\ref{eqn: Psi1}) that $\frac{\nu \log
\nu^{-1}}{(1-\nu)^{2}}=\Psi$, where $\Psi=o(1)$. In this case,
$\nu=o(1)$, and therefore, $\nu \log \nu^{-1} \approx \Psi$. The
solution for this equation is $\nu \approx \frac{\Psi}{\log
(\Psi)^{-1}}$. In other words,
\begin{equation}\label{result07}
e^{-\lambda q_{n}} \approx \dfrac{\frac{n \hat{\alpha}}{\lambda \tau_{n}^{2}}}{\log \left(\frac{\lambda \tau_{n}^{2}}{n \hat{\alpha}} \right)}.
\end{equation}
Thus,
\begin{eqnarray}
\lambda q_{n} &\approx& \log \left(\frac{\lambda \tau_{n}^{2}}{n \hat{\alpha}} \right)+\log \log \left(\frac{\lambda \tau_{n}^{2}}{n \hat{\alpha}} \right)\\
\label{result06}& \stackrel{(a)}{\approx} & \log \left(\frac{\lambda \tau_{n}^{2}}{n \hat{\alpha}} \right),
\end{eqnarray}
where $(a)$ follows from $\lambda q_n = \omega (1)$ which comes from
$\nu =o(1)$. The solution for the above equation can be written as
$\tau_{n}=\log \lambda -f(\lambda)$ or
$q_{n}=\frac{e^{f(\lambda)}}{\lambda}=o(1)$, where we assume
$f(\lambda)=o(\log \lambda)$.  Substituting in (\ref{result06}), we
obtain
\begin{eqnarray}
e^{f(\lambda)}&=&\log \left(\frac{\lambda (\log \lambda -f(\lambda))^{2}}{n \hat{\alpha}} \right)\\
&=&\log \left(\frac{\lambda \log^{2} \lambda}{n \hat{\alpha}} \right)+2\log \left( 1-\dfrac{f(\lambda)}{\log \lambda} \right)\\
&\stackrel{(a)}{\approx}& \log \left(\frac{\lambda \log^{2} \lambda}{n \hat{\alpha}} \right),
\end{eqnarray}
where $(a)$ follows from the fact $f(\lambda)=o(\log \lambda)$.
Thus, using $\tau_{n}=\log \lambda -f(\lambda)$, it yields
\begin{equation}\label{result08}
\tau^{CAP}_{n}=\log \lambda -\log \log \left(\frac{\lambda \log^{2} \lambda}{n \hat{\alpha}} \right).
\end{equation}
It should be noted that (\ref{result08}) is derived from
(\ref{result07}) for $\Psi \triangleq \frac{n
\hat{\alpha}}{\tau_{n}^{2} \lambda} \ll 1$. This translates the
condition $\frac{n \hat{\alpha}}{\tau_{n}^{2} \lambda} \ll 1$ to
$\frac{n \hat{\alpha}}{\lambda \log^{2} \lambda } \ll 1$, which
incurs that $\lambda=\omega \left( \frac{n}{\log^{2}n} \right)$.

Also, in the following we show that the condition $\frac{\lambda
q_{n}^{2}}{2}=o(1)$ is satisfied. It follows from (\ref{result06})
that
\begin{eqnarray}
\lambda q^{2}_{n}&=&\dfrac{\log^{2} \left(\frac{\lambda \tau_{n}^{2}}{n \hat{\alpha}} \right)}{\lambda}\\
&\stackrel{(a)}{\leq}& \dfrac{\log^{2} \left(\frac{\lambda \log^{2} \lambda}{n \hat{\alpha}} \right)}{\lambda}\\
&\stackrel{(b)}{=}& o(1),
\end{eqnarray}
where $(a)$ follows from (\ref{result08}) and $(b)$ comes from
$\lambda=\omega \left( \frac{n}{\log^{2}n} \right)$.

To satisfy the condition of Lemma \ref{lemma006}, we must have
\begin{equation}\label{condi001}
 \dfrac{\tau_n}{n \hat{\alpha} q^{CAP}_n \Delta^{CAP}_{n}} \ll 1.
\end{equation}
From (\ref{CAP001}) and (\ref{result08}), we can write
\begin{eqnarray}
\dfrac{\tau_n}{n \hat{\alpha} q^{CAP}_n \Delta^{CAP}_{n}} &\approx& \dfrac{\lambda \log \lambda}{n \hat{\alpha}\left[ 1-e^{-\lambda q_{n}} \right]}\\
&\stackrel{(a)}{\approx}&\dfrac{\lambda \log \lambda}{n \hat{\alpha}},
\end{eqnarray}
where $(a)$ follows from $e^{-\lambda q_{n}}=o(1)$. In order to have
$\frac{\lambda \log \lambda}{n \hat{\alpha}}=o(1)$, one must have
$\lambda=o \left(\frac{n}{\log n} \right)$. In this case, the
maximum effective throughput of the network can be simplified as
\begin{equation}
\mathfrak{T}_{\mathrm{eff}} \approx \dfrac{\tau_{n}}{\hat{\alpha}} \approx \frac{\log \lambda}{\hat{\alpha}}.
\end{equation}
Noting that $\lambda$ satisfies $\lambda=\omega
\left(\frac{n}{\log^2 n} \right)$ and $\lambda=o \left(\frac{n}{\log
n} \right)$, it follows that $\log \lambda \sim \log n$. In other
words,
 $\mathfrak{T}_{\mathrm{eff}} \approx \frac{\log n}{\hat{\alpha}}$.

\section*{Acknowledgment}
The authors would like to thank V. Pourahmadi of CST Lab. for the
helpful discussions.


\begin{thebibliography}{10}

\bibitem{JamshidISIT2}
J.~Abouei, A.~Bayesteh, and A.~K. Khandani,
\newblock ``Delay-throughput analysis in decentralized single-hop wireless
  networks,''
\newblock in {\em Proc. IEEE International Symposium on Information Theory
  (ISIT'07)}, Nice, France, June 2007, pp. 1401--1405.

\bibitem{FoschiniITVT1193}
G.~J. Foschini and Z.~Miljanic,
\newblock ``A simple distributed autonomous power control algorithm and its
  convergence,''
\newblock {\em IEEE Trans. on Vehicular Technology}, vol. 42, no. 4, pp.
  641--646, Nov. 1993.

\bibitem{YatesJSAC0995}
R.~Yates,
\newblock ``A framework for uplink power control in cellular radio systems,''
\newblock {\em IEEE Journal on Selected Areas in Commun.}, vol. 13, no. 7, pp.
  1341--1348, Sept. 1995.

\bibitem{SaraydarITC0202}
C.~U. Saraydar, N.~B. Mandayam, and D.~J. Goodman,
\newblock ``Efficient power control via pricing in wireless data networks,''
\newblock {\em IEEE Trans. on Commun.}, vol. 50, no. 2, pp. 291--303, Feb.
  2002.

\bibitem{HuangJSAC2006}
J.~Huang, R.~A. Berry, and M.~L. Honig,
\newblock ``Distributed interference compensation for wireless networks,''
\newblock {\em IEEE Journal on Selected Areas in Commun.}, vol. 24, no. 5, pp.
  1074--1084, May 2006.

\bibitem{EtkinTse2007}
R.~Etkin, A.~Parekh, and D.~Tse,
\newblock ``Spectrum sharing for unlicensed bands,''
\newblock {\em IEEE Journal on Selected Areas in Commun.}, vol. 25, no. 3, pp.
  517--528, April 2007.

\bibitem{JindalAllerton07}
N.~Jindal, S.~Weber, and J.~Andrews,
\newblock ``Fractional power control for decentralized wireless networks,''
\newblock {\em IEEE Trans. on Wireless Commun.}, vol. 7, no. 12, pp.
  5482--5492, Dec. 2008.

\bibitem{Jung_Shah2007}
K.~Jung and D.~Shah,
\newblock ``Low delay scheduling in wireless networks,''
\newblock in {\em Proc. IEEE International Symposium on Information Theory
  (ISIT'07)}, June 2007, pp. 1396--1400.

\bibitem{GuptaITIT2000}
P.~Gupta and P.~R. Kumar,
\newblock ``The capacity of wireless networks,''
\newblock {\em IEEE Trans. on Information Theory}, vol. 46, no. 2, pp.
  388--404, March 2000.

\bibitem{GrossglauserIACM0802}
M.~Grossglauser and D.~Tse,
\newblock ``Mobility increases the capacity of ad-hoc wireless networks,''
\newblock {\em IEEE/ACM Trans. on Networking}, vol. 10, no. 4, pp. 477--486,
  August 2002.

\bibitem{KulkaraniITIT0604}
S.~R. Kulkarni and P.~Viswanath,
\newblock ``A deterministic approach to throughput scaling in wireless
  networks,''
\newblock {\em IEEE Trans. on Information Theory}, vol. 50, no. 6, pp.
  1041--1049, June 2004.

\bibitem{XieITIT0504}
L.-L. Xie and P.~R. Kumar,
\newblock ``A network information theory for wireless communication: scaling
  laws and optimal operation,''
\newblock {\em IEEE Trans. on Information Theory}, vol. 50, no. 5, pp.
  748--767, May 2004.

\bibitem{HassibiITIT0706}
R.~Gowaikar, B.~Hochwald, and B.~Hassibi,
\newblock ``Communication over a wireless network with random connections,''
\newblock {\em IEEE Trans. on Information Theory}, vol. 52, no. 7, pp.
  2857--2871, July 2006.

\bibitem{ElgamalITIT0606}
A.~El Gamal, J.~Mammen, B.~Prabhakar, and D.~Shah,
\newblock ``Optimal throughput-delay scaling in wireless networks - part {I}:
  The fluid model,''
\newblock {\em IEEE Trans. on Information Theory}, vol. 52, no. 6, pp.
  2568--2592, June 2006.

\bibitem{LeAlfaITWC1106}
L.~B. Le, E.~Hossain, and T.~S. Alfa,
\newblock ``Delay statistics and throughput performance for multi-rate wireless
  networks under multiuser diversity,''
\newblock {\em IEEE Trans. on Wireless Commun.}, vol. 5, no. 11, pp.
  3234--3243, Nov. 2006.

\bibitem{BetteshPIMRC98}
I.~Bettesh and S.~Shamai,
\newblock ``A low delay algorithm for the multiple access channel with
  {R}ayleigh fading,''
\newblock in {\em Proc. IEEE Personal, Indoor and Mobile Radio Commun.}, Sept.
  1998, vol.~3, pp. 1367--1372.

\bibitem{BansalINFOCOM2003}
N.~Bansal and Z.~Liu,
\newblock ``Capacity, delay and mobility in wireless ad-hoc networks,''
\newblock in {\em Proc. IEEE INFOCOM}, April 2003, pp. 1553--1563.

\bibitem{ToumpisINFOCOM04}
S.~Toumpis and A.~J. Goldsmith,
\newblock ``Large wireless networks under fading, mobility, and delay
  constraints,''
\newblock in {\em Proc. IEEE INFOCOM}, March 2004, pp. 609--619.

\bibitem{GopalaWNCMC2005}
P.~K. Gopala and H.~El Gamal,
\newblock ``On the throughput-delay tradeoff in cellular multicast,''
\newblock in {\em Proc. IEEE International Conference on Wireless Networks,
  Communications and Mobile Computing}, June 2005, vol.~2, pp. 1401--1406.

\bibitem{SharmaITIT0606}
L.~Xiaojun, G.~Sharma, R.~R. Mazumdar, and N.~B. Shroff,
\newblock ``Degenerate delay-capacity tradeoffs in ad-hoc networks with
  brownian mobility,''
\newblock {\em IEEE Trans. on Information Theory}, vol. 52, no. 6, pp.
  2777--2784, June 2006.

\bibitem{NeelyITIT0605}
M.~J. Neely and E.~Modiano,
\newblock ``Capacity and delay tradeoffs for ad-hoc mobile networks,''
\newblock {\em IEEE Trans. on Information Theory}, vol. 51, no. 6, pp.
  1917--1937, June 2005.

\bibitem{SahrifITWC0907}
M.~Sharif and B.~Hassibi,
\newblock ``Delay considerations for opportunistic scheduling in broadcast
  fading channels,''
\newblock {\em IEEE Trans. on Wireless Commun.}, vol. 6, no. 9, pp. 3353--3363,
  Sept. 2007.

\bibitem{AlirezaITIT2007}
A.~Bayesteh, M.~Ansari, and A.~K. Khandani,
\newblock ``Throughput and fairness maximization in wireless downlink
  systems,''
\newblock {\em Submitted to IEEE Trans. on Information Theory}, Sept. 2007.

\bibitem{ComaniciuITWC0806}
C.~Comaniciu and H.~V. Poor,
\newblock ``On the capacity of mobile ad-hoc networks with delay constraints,''
\newblock {\em IEEE Trans. on Wireless Commun.}, vol. 5, no. 8, pp. 2061--2071,
  August 2006.

\bibitem{WangISIT2008}
Z.~Wang, H.~R. Sadjadpour, and J.~J. Garcia-Luna-Aceves,
\newblock ``Capacity-delay tradeoff for information dissemination modalities in
  wireless networks,''
\newblock in {\em Proc. IEEE International Symposium on Information Theory
  (ISIT'08)}, Toronto, Canada, July 2008, pp. 677--681.

\bibitem{WalshISIT2008}
J.~M. Walsh, S.~Weber, and C.~wa~Maina,
\newblock ``Optimal rate delay tradeoffs for multipath routed and network coded
  networks,''
\newblock in {\em Proc. IEEE International Symposium on Information Theory
  (ISIT'08)}, Toronto, Canada, July 2008, pp. 682--686.

\bibitem{JamshidITIT2008}
J.~Abouei, A.~Bayesteh, M.~Ebrahimi, and A.~K. Khandani,
\newblock ``On the throughput maximization in decentralized wireless
  networks,''
\newblock {\em Submitted to IEEE Trans. on Information Theory}, Oct. 2008.

\bibitem{marzetta}
T.~Marzetta and B.~Hochwald,
\newblock ``Capacity of a mobile multiple-antenna communication link in
  rayleigh flat fading,''
\newblock {\em IEEE Trans. on Information Theory}, vol. 45, pp. 139--157, Jan.
  1999.

\bibitem{shlomo}
S.~Shamai and T.~Marzetta,
\newblock ``Multiuser capacity in block fading with no channel state
  information,''
\newblock {\em IEEE Trans. on Information Theory}, vol. 48, pp. 938--942, Aug.
  2002.

\bibitem{Neelyallerton06}
M.~J. Neely,
\newblock ``Order optimal delay for opportunistic scheduling in multi-user
  wireless uplinks and downlinks,''
\newblock in {\em Proc. of 44th Allerton Conference on Communication, Control,
  and Computing}, Sept. 2006.

\bibitem{KnuthACM67}
D.~E. Knuth,
\newblock ``Big omicron and big omega and big theta,''
\newblock in {\em ACM SIGACT News}, April-June 1967, vol.~8, pp. 18--24.

\bibitem{AbdiVTC1999}
A.~Abdi and M.~Kaveh,
\newblock ``On the utility of gamma {PDF} in modeling shadow fading (slow
  fading),''
\newblock in {\em Proc. IEEE Vehicular Technology Conference (VTC'99)}, May
  1998, vol.~3, pp. 2308--2312.

\bibitem{AbdiVTC2001}
A.~Abdi, H.~A. Barger, and M.~Kaveh,
\newblock ``A simple alternative to the lognormal model of shadow fading in
  terrestrial and satellite channels,''
\newblock in {\em Proc. IEEE Vehicular Technology Conference (VTC'01)}, Fall
  2001, vol.~4, pp. 2058--2062.

\bibitem{Gallagerbook95}
R.~G. Gallager,
\newblock {\em Discrete Stochastic Processes},
\newblock Kluwer Academic Publishers, 1995.

\bibitem{ZorziAllerton98}
M.~Zorzi and R.~R. Rao,
\newblock ``On channel modeling for delay analysis of packet communications
  over wireless links,''
\newblock in {\em Proc. 36th Annual Allerton Conference}, Sept. 1998.

\bibitem{TangITWC1207}
J.~Tang and X.~Zhang,
\newblock ``Cross-layer modeling for quality of service guarantees over
  wireless links,''
\newblock {\em IEEE Trans. on Wireless Commun.}, vol. 6, no. 12, pp.
  4504--4512, Dec. 2007.

\bibitem{Petrovbook95}
Valentin~V. Petrov,
\newblock {\em Limit Theorems of Probability Theory: Sequences of Indpendent
  Random Variables},
\newblock Oxford University Press, 1995.

\end{thebibliography}

\end{document}